\documentclass[onecolumn,journal]{IEEEtran}

\usepackage[T1]{fontenc}
\usepackage[latin9]{inputenc}
\usepackage{color}
\usepackage{array}
\usepackage{float}
\usepackage{mathrsfs}
\usepackage{mathtools}
\usepackage{amsthm}
\usepackage{amstext}
\usepackage{amssymb}
\usepackage{stmaryrd}
\usepackage{graphicx}
\usepackage{pgfplots}

\makeatletter

\theoremstyle{plain}
\newtheorem{theorem}{\protect\theoremname}
\theoremstyle{plain}

\theoremstyle{plain}
\newtheorem{corollary}[theorem]{\protect\corollaryname}
\theoremstyle{plain}
\newtheorem{lemma}[theorem]{\protect\lemmaname}
\theoremstyle{definition}
\newtheorem{example}[theorem]{\protect\examplename}
\theoremstyle{definition}

\theoremstyle{definition}
\newtheorem{remark}{\protect\remarkname}

\AtBeginDocument{
	
}

\makeatother

  \providecommand{\corollaryname}{Corollary}
  \providecommand{\examplename}{Example}
  \providecommand{\lemmaname}{Lemma}
  \providecommand{\propositionname}{Proposition}
  \providecommand{\theoremname}{Theorem}
  \providecommand{\definitionname}{Definition}
  \providecommand{\remarkname}{Remark}

\newcommand{\Bk}[1]{{\Big( {#1} \Big) }}
\newcommand{\bk}[1]{{\big( {#1} \big) }}

\newcommand{\ceil}[1]{{\left\lceil {#1} \right\rceil}}

\begin{document}

%
\title{A Construction of Linear Codes and Their Complete Weight Enumerators}
%
%
%

\author{Shudi Yang, Xiangli Kong, Chunming Tang    
\thanks{S. Yang is 
	with the School of Mathematical Sciences, Qufu Normal University, Shandong 273165, P.R.China.\protect\\
\quad X. Kong is with the School of Mathematical
	Sciences, Qufu Normal University, Shandong 273165, P.R.China. \protect\\
\quad C. Tang is with School of Mathematics and Information Science, Guangzhou University, Guangzhou 510006, P.R.China. \protect\\
	\protect\\
	E-mail: yangshudi7902@126.com,~{kongxiangli@126.com},~{ctang@gzhu.edu.cn}.}\protect\\
\thanks{Manuscript received *********; revised ********.}
}

\maketitle

\begin{abstract}
Recently, linear codes constructed from defining sets have been studied extensively. They may have nice parameters if the defining set is chosen properly. Let $ m >2$ be a positive integer. For an odd prime $ p $, let $ r=p^m $ and  $\text{Tr}$ be the absolute trace function from $\mathbb{F}_r$
onto $\mathbb{F}_p$. In this paper, we give a construction of linear codes by defining the code
\begin{align*}
C_{D}=\{(\mathrm{Tr}(ax))_{x\in D}: a \in \mathbb{F}_{r} \},
\end{align*}
where
$
D =\left\{ x\in \mathbb{F}_{r}  :  \mathrm{Tr}(x)=1, \mathrm{Tr}(x^2)=0 \right\}.
$
Its complete weight enumerator and weight enumerator are determined explicitly by employing cyclotomic numbers and Gauss sums. In addition, we obtain several optimal linear codes with a few weights. They have higher rate compared with other codes, which enables them to have essential applications in areas such as association schemes and secret sharing schemes.

\end{abstract}



\IEEEpeerreviewmaketitle

\section{Introduction}
 
\IEEEPARstart{T}{hroughout} this paper, let $p$ be an odd prime, and let $ r=p^m $ for a positive integer $ m >2$. Denote by $\mathbb{F}_r$ a
finite field with $r$ elements. The absolute trace function is denoted by $\mathrm{Tr}$. An $ [n,k,d] $ linear code $ C  $ over $ \mathbb{F}_p $ is a
$ k $-dimensional subspace of $ \mathbb{F}_p^n $ with minimum distance $ d $.
The fraction $ k/n $ is called the rate, or information
rate, and gives a measure of the number of information coordinates relative to the total number of coordinates. The higher the rate, the higher the proportion of coordinates in a codeword actually contain information rather than redundancy (see~\cite{VPless2003funda}). The complete weight enumerator of a code $C$ over $\mathbb{F}_p$, will enumerate the codewords according to the number of symbols of each kind contained in each codeword (see~\cite{macwilliams1977theory}). Denote elements of the field by $\mathbb{F}_p=\{z_0,z_1,\cdots,z_{p-1}\}$, where $z_0=0$.
For a vector $\mathsf{v}=(v_0,v_1,\cdots,v_{n-1})\in \mathbb{F}_p^n$, the composition of $ \mathsf{v} $, denoted by $\text{comp}(\mathsf{v})$, is defined as
$$\text{comp}(\mathsf{v})=(k_0,k_1,\cdots,k_{p-1}),$$
where $k_j$ is the number of components $v_i~(0 \leqslant i \leqslant n-1)$ of $\mathsf{v} $ that equal to $z_j$. It is easy to see that $\sum_{j=0}^{p-1}k_j=n$.
Let $ A(k_0,k_1,\cdots,k_{p-1}) $ be the number of codewords
$\mathsf{c} \in C$ with $\text{comp}(\mathsf{c})=(k_0,k_1,\cdots,k_{p-1})$.
Then the complete weight enumerator of the code $C$ is the polynomial
\begin{align*}
\mathrm{CWE}(C)
&=\sum_{\mathsf{c}\in C}z_0^{k_0}z_1^{k_1}\cdots z_{p-1}^{k_{p-1}}\\
&=\sum_{(k_0,k_1,\cdots,k_{p-1})\in B_n}A(k_0,k_1,\cdots,k_{p-1})z_0^{k_0}z_1^{k_1}\cdots z_{p-1}^{k_{p-1}},
\end{align*}
where $ B_n=\{(k_0,k_1,\cdots,k_{p-1}): 0 \leqslant k_j \leqslant n, \sum_{j=0}^{p-1}k_j=n \} $.
One sees that the key to determining $  \mathrm{CWE}(C) $ of a code $ C $ is determining those $ \text{comp}(\mathsf{c}) $ and $ A(k_0,k_1,\cdots,k_{p-1}) $
such that $ A(k_0,k_1,\cdots,k_{p-1}) \neq 0 $.

The complete weight enumerators of linear codes have been of fundamental importance to theories and practices since they not only give the weight enumerators but also demonstrate the frequency of each symbol appearing in each codeword. Blake and Kith investigated the complete weight enumerator of Reed-Solomon codes and showed that they could be helpful in soft decision decoding~\cite{Blake1991,kith1989complete}. Kuzmin and Nechaev studied the
generalized Kerdock code and related linear codes over Galois rings and estimated their complete weight enumerators in~\cite{kuzmin1999complete} and~\cite{kuzmin2001complete}. Nebe $et~al.$~\cite{Nebe2004} described the complete weight enumerators of generalized doubly-even self-dual codes.    In~\cite{helleseth2006}, the study of the monomial and quadratic bent functions was related to the complete weight enumerators of linear codes. Recently, a lot of progress has been made on this subject. Ding $et~al.$~\cite{ding2007generic,Ding2005auth} showed that complete weight enumerators can be applied to the calculation of the deception probabilities of certain authentication codes. In~\cite{chu2006constantco,ding2008optimal,ding2006construction}, the authors studied the complete weight enumerators of some constant
composition codes and presented some families of optimal constant composition codes.

We introduce the the generic construction of linear codes developed by Ding $et~al.$ in~\cite{ding2015twodesign,dingkelan2014binary,ding2015twothree}. Set $D=\{d_1,d_2,\cdots,d_n\}\subseteq \mathbb{F}_{r}$, where $r=p^m$. A linear code associated with $D$ is defined by
\begin{equation}\label{def:CD}
C_{D}=\{(\mathrm{Tr}(ax))_{x\in D}:
a\in \mathbb{F}_{r}\}.
\end{equation}
Then $D$ is called the defining set of this code $C_{D}$. In~\cite{ding2015twothree}, the authors constructed the code $C_{D}$ with two or three weights whose defining set is $D=\{x\in \mathbb{F}_{r}^*:\mathrm{Tr}(x^{2})=0\}$, and its complete weight enumerator was established in~\cite{LiYang2015cwe,yang2015linear}. Along this inspired idea, many new results are dedicated to computing the complete weight enumerators and weight enumerators of specific codes, see~\cite{LiYang2015cwe,AhnKaLi2016completegenelize,li2015complete,WangQiuyan2015complet,wang2015class,yang2016compthree,yang2015complete,Yang2016complete,Lifei15wt,Heng2016}. All of these researches are concerning the defining set with only one trace function. If we restrict the defining set with two or more trace functions, then it is possible to obtain linear codes with higher rate compared with others.

In this paper, we define the defining set
\begin{align*}
D =\left\{ x\in \mathbb{F}_{r}  :  \mathrm{Tr}(x)=1, \mathrm{Tr}(x^2)=0 \right\},
\end{align*}
and investigate the corresponding code $ C_{D} $ of~\eqref{def:CD}. To be precise, we present explicitly its complete weight enumerator and weight enumerator. Besides, we obtain several optimal linear codes with respect
to the Griesmer bound. We show that they have higher rate compared with other codes so that they have many applications in association
schemes~\cite{calderbank1984three} and secret sharing schemes~\cite{ding2015twothree}.

The organization of this paper is as follows. Section 2 briefly recalls
some definitions and results on cyclotomic numbers and Gauss sums over finite
fields. Section 3 is devoted to the complete weight enumerator and weight enumerator of $ C_D $. We provide some examples to illustrate our main results.
Finally, Section 4 concludes this paper and makes some remarks on this topic.

\section{Mathematical foundations}\label{sec:mathtool}

We begin with some preliminaries by introducing the concept of cyclotomic numbers and Gauss sums over finite fields. Recall that $r=p^m$.
Let $\alpha$ be a primitive element
of $\mathbb{F}_r$ and
$r= Nh+1$ for two positive integers $N>1$, $h>1$.
The \emph{cyclotomic classes} of order $N$
in $\mathbb{F}_r$ are the cosets $C_i^{(N,r)}=\alpha
^i\langle\alpha^N\rangle$ for $i=0,1,\cdots,N-1$, where
$\langle\alpha^N\rangle$ denotes the subgroup of $\mathbb{F}_r^*$
generated by $ \alpha^N$. We know that $\#C_i^{(N,r)}=h$.

For fixed $ i $ and $ j $, we define the \emph{cyclotomic number} $ (i,j)^{(N,r)} $
to be the number of solutions of the equation
\begin{equation*}
x_i+1=x_j ~~~ (x_i \in C_i^{(N,r)} , x_j \in C_j^{(N,r)}  ),
\end{equation*}
where $ 1=\alpha^0 $ is the multiplicative unit of $ \mathbb{F}_{r}  $. That is,
$ (i,j)^{(N,r)} $ is the number of ordered pairs $ (s,t) $ such that
\begin{equation*}
\alpha^{Ns+i}+1=\alpha^{Nt+j} ~~~ ( 0\leqslant s,t \leqslant h-1).
\end{equation*}

If $\lambda$ is a multiplicative character of $\mathbb{F}_r^*$, then we can define the Gauss sum $G(\lambda)$ over $\mathbb{F}_{r}$ as
\begin{equation*}
G(\lambda)=\sum_{x\in\mathbb{F}_r^*}\lambda(x)\zeta_p^{\text{Tr}(x)}.
\end{equation*}
Let $ \eta_m $ denote the quadratic character of $ \mathbb{F}_r $ by defining $ \eta_m(0)=0 $. The quadratic Gauss sum $ G( \eta_m) $ over $ \mathbb{F}_r $
is denoted by $ G_m $. When $ m=1 $, we briefly write $ G( \eta) $ as $ G $, where $ \eta:=\eta_1 $ is the quadratic character over $ \mathbb{F}_p $.

Next, let us review some results on cyclotomic numbers and Gauss sums.
\begin{lemma}$ \cite{storer1967cyclotomy} $ \label{lemN=2}
	When $N=2$, the cyclotomic numbers are given by\\
	$ (1) $ $ h $ even: $(0,0)^{(2,r)} =\frac{h-2}{2} $,
	$ (0,1)^{(2,r)} =(1,0)^{(2,r)} =(1,1)^{(2,r)} = \frac{h}{2}$.\\
	$ (2) $ $ h $ odd: $(0,0)^{(2,r)} =(1,0)^{(2,r)} =(1,1)^{(2,r)} =\frac{h-1}{2} $,
	$ (0,1)^{(2,r)} = \frac{h+1}{2}$.
\end{lemma}

\begin{lemma}$ \cite{lidl1983finite} $ \label{lm:gauss sum}
	Let $ \eta_m $ be the quadratic character of $\mathbb{F}_r$, where $  r = p^m  $, $ m \geqslant 1$. Then
	\begin{equation*}
	G_m=(-1)^{m-1}(-1)^{\frac{(p-1)m}{4}}p^{\frac{m}{2}}.
	\end{equation*}
	In particular, $ G= (-1)^{\frac{p-1}{4}}p^{\frac{1}{2}} $ and $ G^2=\eta(-1)p $.
	
\end{lemma}

\begin{lemma}$ \cite{lidl1983finite} $\label{lm:expo sum}
	Let $r=p^m$ and $f(x)=a_2x^2+a_1x+a_0\in \mathbb{F}_{r}[x]$ with
	$a_2\neq0$.	Then
	\begin{equation*}
	\sum_{x\in
		\mathbb{F}_{r}}\zeta_p^{\mathrm{Tr}(f(x))}=\zeta_p^{\mathrm{Tr}(a_0-a_1^2(4a_2)^{-1})}\eta_m(a_2)G(\eta),
	\end{equation*}
	where $ \eta_m $ is the quadratic character of $ \mathbb{F}_{r} $.
\end{lemma}

The following is the well-known Griesmer bound
(see~\cite{Griesmer1960}) for linear codes over finite fields.
\begin{lemma}$ \cite{Griesmer1960} $ \label{lem:Grie} \textup{(Griesmer Bound)}  Let $ C $ be an $ [n, k, d] $ linear code over $ \mathbb{F}_q $ with $ k \geqslant 1 $ and $ q $ is a power of $ p $. Then
	\begin{align*}
	n \geqslant \sum_{i=0}^{k-1} \ceil{\frac{d}{q^i}} ,
	\end{align*}
	where the symbol $ \lceil x \rceil $ denotes the smallest integer not less than $ x $.
\end{lemma}

\section{Main results}\label{sec:main results}

In this section, we will focus our attention on the complete weight enumerator of $C_{D}$ defined by~\eqref{def:CD}, where
\begin{equation*}
D=\left\{ x\in \mathbb{F}_{r}  :  \mathrm{Tr}(x)=1, \mathrm{Tr}(x^2)=0 \right\}.
\end{equation*}

Now we contribute to determine the parameters of $ C_D $. It is obvious that the length $n$ is equal to the cardinality $ \# D $, which is given in the following Lemma. For later use, we write $ m_p :=m \mod{p} $ for simplicity.
\begin{lemma} $ \cite{Lifei15wt} $  \label{lem:codelength}
	For $ B \in \mathbb{F}_p $, define
	\begin{equation*}
	N(0,B) := \#\left\{x\in \mathbb{F}_{r}  :  \mathrm{Tr}(x^2)=0,\mathrm{Tr}(x)=B  \right\}.
	\end{equation*}
	The following assertions hold.\\
	$ (1) $ If $ B=0 $, then we have
	\begin{align*}
	N(0,0)=\left\{
	\begin{array}{lll}
	p^{m-2}+p^{-1}(p-1)G_m & &\textup{ if } 2\mid m,~ m_p=0,\\
	p^{m-2}	& &\textup{ if } 2\mid m,~ m_p \neq 0,\\
	p^{m-2} & & \textup{ if } 2\nmid m,~ m_p=0,\\
	p^{m-2}+p^{-2}\eta(-m_p)(p-1)G_m G & & \textup{ if } 2\nmid m,~ m_p\neq 0.
	\end{array} \right.
	\end{align*}
	$ (2) $ If $ B\neq 0 $, then we have
	\begin{align*}
	N(0,B)=\left\{
	\begin{array}{lll}
	p^{m-2}      & &\textup{ if } 2\mid m,~ m_p=0,\\
	p^{m-2}+p^{-1} G_m	& &\textup{ if } 2\mid m,~ m_p \neq 0,\\
	p^{m-2}        & & \textup{ if } 2\nmid m,~ m_p=0,\\
	p^{m-2}-p^{-2}\eta(-m_p) G_m G \phantom{p-1} & & \textup{ if } 2\nmid m,~ m_p\neq 0.
	\end{array} \right.
	\end{align*}
\end{lemma}

It follows immediately from the previous lemma that the length of $ C_D $ is $ n=N(0,1) $.

Let $ \rho \in \mathbb{F}_p^* $ and $ a \in \mathbb{F}_r^* $. For a
codeword $ \mathsf{c}(a ) $ of $ C_D $,
we denote $  N_{\rho} :=N_{\rho}(a) $ to be the number of components $ \textup{Tr}(ax) $ of
$ \mathsf{c}(a  ) $ that are equal to $ \rho $. Then
\begin{align}\label{eq:N_rho}
N_{\rho}
&=\#\{x\in \mathbb{F}_r :\mathrm{Tr}(x )=1,\mathrm{Tr}(x^2 )=0,\mathrm{Tr}(ax)=\rho\} \nonumber \\
&=\sum_{ x \in \mathbb{F}_r}
\Big( \dfrac{1}{p}\sum_{y\in \mathbb{F}_p}\zeta_p^{y(\mathrm{Tr}(x )-1)}\Big)  \Big( \dfrac{1}{p}\sum_{z\in \mathbb{F}_p}\zeta_p^{z\mathrm{Tr}(x^2)}\Big)
\Big(  \dfrac{1}{p}\sum_{\delta \in \mathbb{F}_p}\zeta_p^{\delta(\mathrm{Tr}(ax)-\rho)} \Big) \nonumber \\
&=\dfrac{n}{p}+
p^{-3} \sum_{ x \in \mathbb{F}_r}
\sum_{y\in \mathbb{F}_p}\zeta_p^{y(\mathrm{Tr}(x )-1)}
\sum_{z\in \mathbb{F}_p}\zeta_p^{z\mathrm{Tr}(x^2)}
\sum_{\delta \in \mathbb{F}_p^*}\zeta_p^{\delta(\mathrm{Tr}(ax)-\rho)}  \nonumber  \\
&= \dfrac{n}{p}+p^{-3}(\Omega_1+\Omega_2+\Omega_3+\Omega_4),
\end{align}
where	
\begin{align*}	
\Omega_1&=\sum_{ x \in \mathbb{F}_r}
\sum_{\delta \in \mathbb{F}_p^*}
\zeta_p^{\delta (\mathrm{Tr}(ax)-\rho)}
=\sum_{\delta \in \mathbb{F}_p^*}  \zeta_p^{-\rho \delta }
\sum_{ x \in \mathbb{F}_r}\zeta_p^{\mathrm{Tr}(a\delta x)}=0,\\
\Omega_2&=\sum_{x\in \mathbb{F}_r}
\sum_{y\in \mathbb{F}_p^*}\zeta_p^{y(\mathrm{Tr}(x )-1)}
\sum_{\delta \in \mathbb{F}_p^*} \zeta_p^{\delta (\mathrm{Tr}(ax)-\rho)},\\
\Omega_3&=\sum_{ x \in \mathbb{F}_r}
\sum_{z\in \mathbb{F}_p^*}\zeta_p^{z\mathrm{Tr}(x^2)}
\sum_{\delta \in \mathbb{F}_p^*} \zeta_p^{\delta (\mathrm{Tr}(ax)-\rho)},\\
\Omega_4&=\sum_{ x \in \mathbb{F}_r}
\sum_{y\in \mathbb{F}_p^*}\zeta_p^{y(\mathrm{Tr}(x )-1)}
\sum_{z\in \mathbb{F}_p^*}\zeta_p^{z\mathrm{Tr}(x^2)}
\sum_{\delta \in \mathbb{F}_p^*} \zeta_p^{\delta (\mathrm{Tr}(ax)-\rho)}.
\end{align*}

We are going to determine the values of $ \Omega_2,\Omega_3 $ and $ \Omega_4 $ in Lemmas~\ref{lem:omega2},~\ref{lem:omega3} and~\ref{lem:omega4}. For convenience, we denote $ A:= \mathrm{Tr}(a^2)$ and $ B:=\mathrm{Tr}(a) $.
\begin{lemma}\label{lem:omega2}
	For $ a \in \mathbb{F}_r^* $ and $ \rho \in \mathbb{F}_p^* $, we have
	\begin{align*}
	\Omega_2=\left\{
	\begin{array}{lll}
	(p-1)r & &\textup{ if } a \in\mathbb{F}_p^*,  \rho=a,\\
	-r	& &\textup{ if } a \in\mathbb{F}_p^*, \rho \neq a,\\
	0   & & \textup{ otherwise.}
	\end{array} \right.
	\end{align*}
\end{lemma}
\begin{proof}
	It follows from the definition that
	\begin{equation*}
	\Omega_2 =
	\sum_{y\in \mathbb{F}_p^*}
	\sum_{\delta \in \mathbb{F}_p^*} \zeta_p^{-y-\rho \delta}
	\sum_{x\in \mathbb{F}_r}\zeta_p^{\mathrm{Tr}\big((a\delta+y)x\big)}.
	\end{equation*}	
	Note that the equation $ a\delta+y=0 $
	has solutions if and only if $ a \in \mathbb{F}_p^* $ for $ y,\delta \in \mathbb{F}_p^* $. Then $ \Omega_2 = 0 $  if $  a \notin \mathbb{F}_p^*  $. Hence, if $ a \in \mathbb{F}_p^* $, we have by the orthogonal property of additive characters that
	\begin{align*}
	\Omega_2  =
	\sum\limits_{\delta \in \mathbb{F}_p^* \atop y=-a\delta}
	\zeta_p^{-y-\rho \delta}
	\sum_{x\in \mathbb{F}_r}\zeta_p^{\mathrm{Tr}\big((a\delta+y)x\big)}
	= r \sum_{\delta \in \mathbb{F}_p^*  }  \zeta_p^{(a-\rho )\delta} .
	\end{align*}	
	The desired conclusion then follows.
\end{proof}

\begin{lemma}\label{lem:omega3}
	For $ a \in \mathbb{F}_r^* $, we have
	\begin{align*}
	\Omega_3=\left\{
	\begin{array}{lll}
	-(p-1)G_m & &\textup{ if } 2\mid m   ,  A=0,\\
	G_m	& &\textup{ if }  2\mid m   ,  A \neq 0,\\
	0   & & \textup{ if }  2\nmid m   ,  A = 0, \\
	-\eta(-A)G_mG   & & \textup{ if }  2\nmid m   ,  A \neq 0.
	\end{array} \right.
	\end{align*}
\end{lemma}
\begin{proof}
	It follows from Lemma~\ref{lm:expo sum} that
	\begin{align*}
	\Omega_3&=\sum_{\delta \in \mathbb{F}_p^*} \zeta_p^{-\rho\delta}
	\sum_{z\in \mathbb{F}_p^*} \sum_{ x \in \mathbb{F}_r} \zeta_p^{ \mathrm{Tr}(zx^2+a\delta x)}  \\
	&= \sum_{\delta \in \mathbb{F}_p^*} \zeta_p^{-\rho\delta}
	\sum_{z\in \mathbb{F}_p^*}  \zeta_p^{ -\frac{\delta^2}{4z} A} \eta_m(z)G_m\\
	&= \left\{
	\begin{array}{lll}
	G_m \sum_{\delta \in \mathbb{F}_p^*} \zeta_p^{-\rho\delta}
	\sum_{z\in \mathbb{F}_p^*}  \zeta_p^{ -\frac{\delta^2}{4z} A}
	& &\textup{ if } 2\mid m ,\\
	G_m \sum_{\delta \in \mathbb{F}_p^*} \zeta_p^{-\rho\delta}
	\sum_{z\in \mathbb{F}_p^*}  \zeta_p^{ -\frac{\delta^2}{4z} A} \eta(z)
	& & \textup{ if }  2\nmid m .
	\end{array} \right.
	\end{align*}
	The desired conclusion then follows.
\end{proof}

\begin{lemma}\label{lem:omega4}
	For $ a \in \mathbb{F}_r^* $ and $ \rho \in \mathbb{F}_p^* $, we have the following assertions. \\
	$ (1) $ If $2 \mid m$ and $ m_p=0 $, then		
	\begin{align*}
	\Omega_4 =  \left\{
	\begin{array}{lll}
	(p-1 ) G_m  	& &\textup{ if } A=0 , B=0,\\
	-G_m         	& & \textup{ if } A=0, B\neq 0,\\
	-G_m         	& & \textup{ if } A\neq 0, B= 0,\\
	(p^2-p-1) G_m   & & \textup{ if } AB\neq 0, A=2\rho B,\\
	-(p+1)G_m       & & \textup{ if } AB\neq 0, A\neq 2 \rho B.
	\end{array} \right.
	\end{align*}
	$ (2) $ If $2 \mid m$ and $ m_p\neq 0 $, then\\
	$ (2.a) $ if $ A=0 $, we have
	\begin{align*}
	\Omega_4 & =  \left\{
	\begin{array}{lll}
	-G_m                & &\textup{ if } B=0,\\
	(p^2-p-1 ) G_m  	& &\textup{ if }  B \neq 0, \rho m_p =2  B,\\
	-(p+1 ) G_m     	& & \textup{ if } B \neq 0, \rho m_p \neq 2  B;
	\end{array} \right.
	\end{align*}
	$ (2.b) $ if $ A\neq 0 $, we have
	\begin{align*}
	\Omega_4   & =  \left\{
	\begin{array}{lll}
	(p^2-p-1 ) G_m  	& &\textup{ if }  \Delta=0, \rho B = A,\\
	-(p+1 ) G_m     	& & \textup{ if }  \Delta=0,\rho B \neq A,\\
	(p^2-2p-1 )  G_m  	& &\textup{ if }  \eta(\Delta)=1,f(\rho)=0,\\
	-(2p+1) G_m     	& & \textup{ if }   \eta(\Delta)=1, f(\rho) \neq 0,\\
	-G_m               & & \textup{ if }   \eta(\Delta)=-1.
	\end{array} \right.
	\end{align*}
	$ (3) $ If $2 \nmid m$ and $ m_p=0 $, then
	\begin{align*}
	\Omega_4
	=  \left\{
	\begin{array}{lll}
	0  	& &\textup{ if } A=0, B=0,\\
	\eta\bk{\frac{\rho B}{2}} p G_m G    	& & \textup{ if } A=0, B \neq 0,\\
	\eta(-A) G_m G      & &\textup{ if } A\neq 0, B=0,\\
	\eta(-A)G_m G  	& &\textup{ if }  AB\neq 0, A=2\rho B,\\
	\Bk{\eta (2\rho B-A)p +\eta(-A)} G_m G      	& & \textup{ if } AB\neq 0, A \neq 2\rho B.
	\end{array} \right.
	\end{align*}
	$ (4) $ If $2 \nmid m$ and $ m_p\neq 0 $, then\\
	$ (4.a) $ if $ A=0 $, we have
	\begin{align*}
	\Omega_4 & =  \left\{
	\begin{array}{lll}
	\eta(-m_p) G_m G                & &\textup{ if } B=0,\\
	\eta(-m_p) G_m G 	& &\textup{ if }  B \neq 0, \rho m_p =2  B,\\
	\Bk{\eta(2B\rho-m_p \rho^2)p+ \eta(-m_p) }G_m G    	& & \textup{ if } B \neq 0, \rho m_p \neq 2  B;
	\end{array} \right.
	\end{align*}
	$ (4.b) $ if $ A\neq 0 $, we have
	\begin{align*}
	\Omega_4   & =  \left\{
	\begin{array}{lll}
	-(p-2 ) \eta(-m_p) G_m G  	& &\textup{ if }  \Delta=0, \rho B = A,\\
	2\eta(-m_p)  G_m G    	& & \textup{ if }  \Delta=0,\rho B \neq A,\\
	\bk{\eta(-A)+\eta(-m_p)} G_m G	& &\textup{ if } \Delta \neq 0,f(\rho)=0,\\
	\Bk{p\eta(f(\rho))+\eta(-A)   +\eta(-m_p)} G_m G   	& &\textup{ if }  \Delta \neq 0,f(\rho)\neq 0 .
	\end{array} \right.
	\end{align*}
\end{lemma}
\begin{proof}
	It follows from Lemma~\ref{lm:expo sum} that
	\begin{align}\label{eq:Omega4}
	\Omega_4 &=\sum_{z\in \mathbb{F}_p^*}
	\sum_{y\in \mathbb{F}_p^*}\zeta_p^{-y}
	\sum_{\delta \in \mathbb{F}_p^*}    \zeta_p^{-\rho\delta}
	\sum_{ x \in \mathbb{F}_r}
	\zeta_p^{ \mathrm{Tr}\big(zx^2+(a\delta+y)x\big) } \nonumber \\
	& = G_m \sum_{z\in \mathbb{F}_p^*} \eta_m(z)
	\sum_{y\in \mathbb{F}_p^*}\zeta_p^{-y}
	\sum_{\delta \in \mathbb{F}_p^*}
	\zeta_p^{- \frac{1}{4z} \mathrm{Tr}\big( (a\delta+y)^2\big)-\rho\delta }\nonumber \\
	& = G_m \sum_{z\in \mathbb{F}_p^*} \eta_m(z)
	\sum_{y\in \mathbb{F}_p^*}\zeta_p^{-\frac{m_p}{4z}y^2-y}
	\sum_{\delta \in \mathbb{F}_p^*}
	\zeta_p^{- \frac{A}{4z} \delta^2 -\big(  \frac{By}{2z}+\rho \big)\delta }.
	\end{align}
	Thus there are four distinct cases to consider:
	\begin{enumerate}
		\item[(1)]  $2 \mid m$ and $ m_p=0 $,
		\item[(2)]  $2 \mid m$ and $ m_p \neq 0 $,
		\item[(3)]  $2 \nmid m$ and $ m_p=0 $,
		\item[(4)]  $2 \nmid m$ and $ m_p\neq 0 $.
	\end{enumerate}
	
	Case (1): Suppose that $2 \mid m$ and $ m_p=0 $.
	We obtain by~\eqref{eq:Omega4} that
	\begin{align*}
	\Omega_4
	= G_m \sum_{z\in \mathbb{F}_p^*}
	\sum_{y\in \mathbb{F}_p^*}\zeta_p^{-y}
	\sum_{\delta \in \mathbb{F}_p^*}
	\zeta_p^{- \frac{A}{4z} \delta^2 -\big(  \frac{By}{2z}+\rho \big)\delta }.
	\end{align*}
	If $ A=0 $, then
	\begin{align*}
	\Omega_4
	& = G_m
	\sum_{y\in \mathbb{F}_p^*}\zeta_p^{-y}
	\sum_{\delta \in \mathbb{F}_p^*}  \zeta_p^{ -\rho\delta}
	\sum_{z\in \mathbb{F}_p^*}  \zeta_p^{ -\frac{By\delta}{2z} } \\
	& =  \left\{
	\begin{array}{lll}
	(p-1 ) G_m  	& &\textup{ if }  B=0,\\
	-G_m         	& & \textup{ if }  B\neq 0.
	\end{array} \right.
	\end{align*}
	If $ A\neq 0 $ and $ B=0 $, then we have from Lemma~\ref{lm:expo sum} that
	\begin{align*}
	\Omega_4
	& =  G_m \sum_{z\in \mathbb{F}_p^*}
	\sum_{y\in \mathbb{F}_p^*}\zeta_p^{-y}
	\sum_{\delta \in \mathbb{F}_p^*}
	\zeta_p^{- \frac{A}{4z}\delta^2-\rho \delta }\\
	& =  G_m  \sum_{z\in \mathbb{F}_p^*}
	\sum_{y\in \mathbb{F}_p^*}\zeta_p^{-y}
	\Bk{ \zeta_p^{  \frac{\rho^2}{A}z } \eta\big(-\frac{A}{4z}\big) G-1 }\\
	& =     G_m G
	\sum_{y\in \mathbb{F}_p^*}\zeta_p^{-y}
	\sum_{z\in \mathbb{F}_p^*}
	\zeta_p^{ \frac{\rho^2}{A}z } \eta\big(-\frac{A}{4z}\big) +(p-1 ) G_m \\
	& = -\eta(-1)G_mG^2+(p-1 ) G_m =-G_m.
	\end{align*}
	If $ A  B\neq 0 $, again from Lemma~\ref{lm:expo sum}, we have
	\begin{align*}
	\Omega_4
	& =  G_m  \sum_{z\in \mathbb{F}_p^*}
	\sum_{y\in \mathbb{F}_p^*}\zeta_p^{-y}
	\Bk{ \zeta_p^{  \frac{B^2}{4Az}y^2+\frac{\rho B}{A}y+\frac{\rho^2}{A}z } \eta\big(-\frac{A}{4z}\big) G-1 } \\
	& =  G_m  G
	\sum_{z\in \mathbb{F}_p^*} \eta\big(-\frac{A}{4z}\big) \zeta_p^{  \frac{\rho^2}{A}z }
	\sum_{y\in \mathbb{F}_p^*}\zeta_p^{\frac{B^2}{4Az}y^2+\bk{\frac{\rho B}{A}-1}y}
	+(p-1 ) G_m \\
	& =     G_m G
	\sum_{z\in \mathbb{F}_p^*}
	\eta\big(-\frac{A}{4z}\big) \zeta_p^{  \frac{\rho^2}{A}z }
	\Bk{
		\zeta_p^{- \frac{\rho^2}{A}z  + \frac{2\rho B-A}{B^2}z}
		\eta\bk{\frac{B^2}{4Az}}G-1} +(p-1 ) G_m \\
	& = \eta(-1)G_mG^2  \sum_{z\in \mathbb{F}_p^*}
	\zeta_p^{  \frac{2\rho B-A}{B^2}z} - \eta(-1)G_mG^2
	+(p-1 ) G_m \\
	& =  \left\{
	\begin{array}{lll}
	\eta(-1)G_mG^2(p-2) +(p-1 ) G_m  	& &\textup{ if }  A=2\rho B\\
	-2\eta(-1)G_mG^2    +(p-1 ) G_m     	& & \textup{ if }  A \neq 2\rho B
	\end{array} \right.\\
	& =  \left\{
	\begin{array}{lll}
	(p^2-p-1 ) G_m  	& &\textup{ if }  A=2\rho B,\\
	-(p+1 ) G_m     	& & \textup{ if }  A \neq 2\rho B.
	\end{array} \right.
	\end{align*}
	
	Case (2): Suppose that $2 \mid m$ and $ m_p\neq 0 $.
	We obtain by~\eqref{eq:Omega4} that
	\begin{align*}
	\Omega_4
	= G_m \sum_{z\in \mathbb{F}_p^*}
	\sum_{y\in \mathbb{F}_p^*}\zeta_p^{-\frac{m_p}{4z}y^2-y}
	\sum_{\delta \in \mathbb{F}_p^*}
	\zeta_p^{- \frac{A}{4z} \delta^2 -\big(  \frac{By}{2z}+\rho \big)\delta }.
	\end{align*}
	If $ A=0 $ and $ B=0 $, then
	\begin{align*}
	\Omega_4
	& = G_m \sum_{z\in \mathbb{F}_p^*}
	\sum_{y\in \mathbb{F}_p^*}\zeta_p^{-\frac{m_p}{4z}y^2-y}
	\sum_{\delta \in \mathbb{F}_p^*}
	\zeta_p^{-  \rho\delta }\\
	& = - G_m \sum_{z\in \mathbb{F}_p^*}
	\sum_{y\in \mathbb{F}_p^*}\zeta_p^{-\frac{m_p}{4z}y^2-y} \\
	& = - G_m \sum_{z\in \mathbb{F}_p^*}
	\Bk{\zeta_p^{\frac{z}{m_p}}\eta\bk{-\frac{m_p}{4z}}G-1 }\\
	& = - \eta(-1)G_m G^2 +(p-1)G_m =-G_m.
	\end{align*}
	If $ A=0 $ and $ B\neq 0 $, then
	\begin{align*}
	\Omega_4
	& = G_m \sum_{z\in \mathbb{F}_p^*}
	\sum_{\delta \in \mathbb{F}_p^*} \zeta_p^{-  \rho\delta }
	\sum_{y\in \mathbb{F}_p^*}\zeta_p^{-\frac{m_p}{4z}y^2-\bk{\frac{B\delta}{2z}+1}y}
	\\
	& =  G_m \sum_{z\in \mathbb{F}_p^*}
	\sum_{\delta \in \mathbb{F}_p^*} \zeta_p^{-  \rho\delta }
	\Bk{ \zeta_p^{ \frac{z}{m_p}\bk{ \frac{B\delta}{2z}+1}^2 }
		\eta\bk{-\frac{m_p}{4z}} G -1}  \\
	& =   G_m G \sum_{z\in \mathbb{F}_p^*}
	\eta\bk{-\frac{m_p}{4z}} \zeta_p^{ \frac{z}{m_p}}
	\sum_{\delta \in \mathbb{F}_p^*}
	\zeta_p^{\frac{B^2}{4zm_p}\delta^2+ \bk{\frac{B}{m_p}-\rho}\delta}+(p-1)G_m \\
	& = G_m G \sum_{z\in \mathbb{F}_p^*}
	\eta\bk{-\frac{m_p}{4z}} \zeta_p^{ \frac{z}{m_p}}
	\Bk{
		\zeta_p^{-\frac{zm_p}{B^2} \bk{\frac{B}{m_p}-\rho}^2 } \eta(zm_p)G-1}
	+(p-1)G_m\\
	& = \eta(-1) G_m G^2  \sum_{z\in \mathbb{F}_p^*}
	\zeta_p^{\frac{2B-\rho m_p}{B^2} \rho z } -\eta(-1) G_m G^2
	+(p-1)G_m \\
	& =  \left\{
	\begin{array}{lll}
	(p-2 )\eta(-1) G_m G^2 +(p-1)G_m	& &\textup{ if } \rho m_p = 2  B\\
	-2\eta(-1) G_m G^2+(p+1 ) G_m     	& & \textup{ if }  \rho m_p \neq 2  B
	\end{array} \right.\\
	& =  \left\{
	\begin{array}{lll}
	(p^2-p-1 ) G_m  	& &\textup{ if }  \rho m_p =2  B,\\
	-(p+1 ) G_m     	& & \textup{ if } \rho m_p \neq 2  B.
	\end{array} \right.
	\end{align*}
	Let $ \Delta := B^2-m_p A  $. If $ A \neq 0 $, then
	\begin{align*}
	\Omega_4
	& = G_m \sum_{z\in \mathbb{F}_p^*}
	\sum_{y\in \mathbb{F}_p^*}\zeta_p^{-\frac{m_p}{4z}y^2-y}
	\Bk{ \zeta_p^{  \frac{B^2}{4Az}y^2+\frac{\rho B}{A}y+\frac{\rho^2}{A}z } \eta\big(-\frac{A}{4z}\big) G-1 }\\
	& = G_m G \sum_{z\in \mathbb{F}_p^*}
	\eta\bk{-\frac{A}{4z}} \zeta_p^{\frac{ \rho^2}{A}z}
	\sum_{y\in \mathbb{F}_p^*}\zeta_p^{\frac{\Delta}{4Az}y^2+\bk{ \frac{\rho B}{A}-1}y}
	- G_m \sum_{z\in \mathbb{F}_p^*}
	\sum_{y\in \mathbb{F}_p^*}\zeta_p^{-\frac{m_p}{4z}y^2-y}.
	\end{align*}
	Therefore, if $ A \neq 0 $ and $ \Delta =0 $, then
	\begin{align*}
	\Omega_4
	& = G_m G \sum_{z\in \mathbb{F}_p^*}
	\eta\bk{-\frac{A}{4z}} \zeta_p^{\frac{ \rho^2}{A}z}
	\sum_{y\in \mathbb{F}_p^*}\zeta_p^{\bk{ \frac{\rho B}{A}-1}y}
	- G_m\\
	& =  \left\{
	\begin{array}{lll}
	(p-1 )\eta(-1) G_m G^2 -G_m	& &\textup{ if } \rho B = A\\
	-\eta(-1) G_m G^2- G_m     	& & \textup{ if }  \rho B \neq A
	\end{array} \right.\\
	& =  \left\{
	\begin{array}{lll}
	(p^2-p-1 )  G_m  	& &\textup{ if } \rho B = A,\\
	-(p+1) G_m     	& & \textup{ if }  \rho B \neq A.
	\end{array} \right.
	\end{align*}
	Suppose that $ A \neq 0 $ and $ \Delta \neq 0 $. Then
	\begin{align*}
	\Omega_4
	& = G_m G \sum_{z\in \mathbb{F}_p^*}
	\eta\bk{-\frac{A}{4z}} \zeta_p^{\frac{ \rho^2}{A}z}
	\Bk{\zeta_p^{-\frac{Az}{\Delta}\bk{ \frac{\rho B}{A}-1}^2}\eta\bk{ \frac{\Delta}{4Az}}G  -1}
	- G_m\\
	& = G_m G^2 \sum_{z\in \mathbb{F}_p^*} \eta(-\Delta) \zeta_p^{f(\rho)\frac{z}{\Delta}}-\eta(-1) G_m G^2 -G_m\\
	& =  G_m G^2 \sum_{z\in \mathbb{F}_p^*} \eta(-\Delta) \zeta_p^{f(\rho)\frac{z}{\Delta}}-(p+1) G_m,
	\end{align*}
	where we denote $ f(\rho):=-m_p \rho^2 +2B\rho-A $.
	Note that the equation $   f(\rho)=0 $ over $ \mathbb{F}_p $ has two distinct solutions if and only if $ \eta(\Delta)=1 $. Therefore, when $ \eta(\Delta)=1 $ we obtain
	\begin{align*}
	\Omega_4 & =  \left\{
	\begin{array}{lll}
	(p-1 )\eta(-1) G_m G^2 -(p+1)G_m	& &\textup{ if } f(\rho)=0\\
	-\eta(-1) G_m G^2 -(p+1) G_m     	& & \textup{ if }  f(\rho) \neq 0
	\end{array} \right.\\
	& =  \left\{
	\begin{array}{lll}
	(p^2-2p-1 )  G_m  	& &\textup{ if }  f(\rho)=0,\\
	-(2p+1) G_m     	& & \textup{ if }   f(\rho) \neq 0.
	\end{array} \right.
	\end{align*}
	Since $ f(\rho) \neq 0 $ as $ \eta(\Delta)=-1 $, we have
	\begin{align*}
	\Omega_4  =  \eta(-1) G_m G^2 -(p+1) G_m
	=- G_m   .
	\end{align*}
	
	Case (3): Suppose that $2 \nmid m$ and $ m_p=0 $.
	We obtain by~\eqref{eq:Omega4} that
	\begin{align*}
	\Omega_4
	= G_m \sum_{z\in \mathbb{F}_p^*} \eta(z)
	\sum_{y\in \mathbb{F}_p^*}\zeta_p^{-y}
	\sum_{\delta \in \mathbb{F}_p^*}
	\zeta_p^{- \frac{A}{4z} \delta^2 -\big(  \frac{By}{2z}+\rho \big)\delta }.
	\end{align*}
	If $ A=0 $, then
	\begin{align*}
	\Omega_4
	& = G_m
	\sum_{y\in \mathbb{F}_p^*}\zeta_p^{-y}
	\sum_{\delta \in \mathbb{F}_p^*}  \zeta_p^{-\rho\delta}
	\sum_{z\in \mathbb{F}_p^*} \eta(z)
	\zeta_p^{  -  \frac{By}{2z}  \delta }\\
	& =  \left\{
	\begin{array}{lll}
	0  	& &\textup{ if }  B=0,\\
	\eta\bk{-\frac{\rho B}{2}} G_m G^3    	& & \textup{ if }  B \neq 0.
	\end{array} \right. \\
	& =  \left\{
	\begin{array}{lll}
	0  	& &\textup{ if }  B=0,\\
	\eta\bk{\frac{\rho B}{2}} p G_m G    	& & \textup{ if }  B \neq 0.
	\end{array} \right.
	\end{align*}
	If $ A\neq 0 $ and $ B=0 $, then we have from Lemma~\ref{lm:expo sum} that
	\begin{align*}
	\Omega_4
	& =  G_m \sum_{z\in \mathbb{F}_p^*} \eta(z)
	\sum_{y\in \mathbb{F}_p^*}\zeta_p^{-y}
	\sum_{\delta \in \mathbb{F}_p^*}
	\zeta_p^{- \frac{A}{4z}\delta^2-\rho \delta }\\
	& =  G_m  \sum_{z\in \mathbb{F}_p^*}  \eta(z)
	\sum_{y\in \mathbb{F}_p^*}\zeta_p^{-y}
	\Bk{ \zeta_p^{  \frac{\rho^2}{A}z } \eta\big(-\frac{A}{4z}\big) G-1 }\\
	& =    \eta(-A) G_m G
	\sum_{y\in \mathbb{F}_p^*}\zeta_p^{-y}
	\sum_{z\in \mathbb{F}_p^*}
	\zeta_p^{ \frac{\rho^2}{A}z }   \\
	& = \eta(-A) G_m G .
	\end{align*}
	If $ A  B\neq 0 $, again from Lemma~\ref{lm:expo sum}, we have
	\begin{align*}
	\Omega_4
	& =  G_m  \sum_{z\in \mathbb{F}_p^*}  \eta(z)
	\sum_{y\in \mathbb{F}_p^*}\zeta_p^{-y}
	\Bk{ \zeta_p^{  \frac{B^2}{4Az}y^2+\frac{\rho B}{A}y+\frac{\rho^2}{A}z } \eta\big(-\frac{A}{4z}\big) G-1 } \\
	& = \eta(-A) G_m  G
	\sum_{z\in \mathbb{F}_p^*}   \zeta_p^{  \frac{\rho^2}{A}z }
	\sum_{y\in \mathbb{F}_p^*}\zeta_p^{\frac{B^2}{4Az}y^2+\bk{\frac{\rho B}{A}-1}y} \\
	& =   \eta(-A)  G_m G
	\sum_{z\in \mathbb{F}_p^*}
	\zeta_p^{  \frac{\rho^2}{A}z }
	\Bk{
		\zeta_p^{- \frac{\rho^2}{A}z  + \frac{2\rho B-A}{B^2}z}
		\eta\bk{\frac{B^2}{4Az}}G-1}   \\
	& = \eta(-1)G_mG^2  \sum_{z\in \mathbb{F}_p^*} \eta(z)
	\zeta_p^{  \frac{2\rho B-A}{B^2}z} + \eta(-A)G_m G \\
	& =  \left\{
	\begin{array}{lll}
	\eta(-A)G_m G 	& &\textup{ if }  A=2\rho B\\
	\eta(-1) \eta (2\rho B-A)G_mG^3    +\eta(-A)G_m G    	& & \textup{ if }  A \neq 2\rho B
	\end{array} \right.\\
	& =  \left\{
	\begin{array}{lll}
	\eta(-A)G_m G  	& &\textup{ if }  A=2\rho B,\\
	\Bk{\eta (2\rho B-A)p +\eta(-A)} G_m G      	& & \textup{ if }  A \neq 2\rho B.
	\end{array} \right.
	\end{align*}
	
	Case (4): Suppose that $2 \nmid m$ and $ m_p \neq 0 $.
	We obtain by~\eqref{eq:Omega4} that
	\begin{align*}
	\Omega_4
	= G_m \sum_{z\in \mathbb{F}_p^*} \eta(z)
	\sum_{y\in \mathbb{F}_p^*}\zeta_p^{-\frac{m_p}{4z}y^2-y}
	\sum_{\delta \in \mathbb{F}_p^*}
	\zeta_p^{- \frac{A}{4z} \delta^2 -\big(  \frac{By}{2z}+\rho \big)\delta }.
	\end{align*}
	If $ A=0 $ and $ B=0 $, then
	\begin{align*}
	\Omega_4
	& = G_m \sum_{z\in \mathbb{F}_p^*}   \eta(z)
	\sum_{y\in \mathbb{F}_p^*}\zeta_p^{-\frac{m_p}{4z}y^2-y}
	\sum_{\delta \in \mathbb{F}_p^*}
	\zeta_p^{-  \rho\delta }\\
	& = - G_m \sum_{z\in \mathbb{F}_p^*}   \eta(z)
	\sum_{y\in \mathbb{F}_p^*}\zeta_p^{-\frac{m_p}{4z}y^2-y} \\
	& = - G_m \sum_{z\in \mathbb{F}_p^*}  \eta(z)
	\Bk{\zeta_p^{\frac{z}{m_p}}\eta\bk{-\frac{m_p}{4z}}G-1 }\\
	& =  \eta(-m_p)G_m G.
	\end{align*}
	If $ A=0 $ and $ B\neq 0 $, then
	\begin{align*}
	\Omega_4
	& = G_m \sum_{z\in \mathbb{F}_p^*}  \eta(z)
	\sum_{\delta \in \mathbb{F}_p^*} \zeta_p^{-  \rho\delta }
	\sum_{y\in \mathbb{F}_p^*}\zeta_p^{-\frac{m_p}{4z}y^2-\bk{\frac{B\delta}{2z}+1}y}
	\\
	& =  G_m \sum_{z\in \mathbb{F}_p^*}  \eta(z)
	\sum_{\delta \in \mathbb{F}_p^*} \zeta_p^{-  \rho\delta }
	\Bk{ \zeta_p^{ \frac{z}{m_p}\bk{ \frac{B\delta}{2z}+1}^2 }
		\eta\bk{-\frac{m_p}{4z}} G -1}  \\
	& =  \eta(-m_p) G_m G \sum_{z\in \mathbb{F}_p^*}
	\zeta_p^{ \frac{z}{m_p}}
	\sum_{\delta \in \mathbb{F}_p^*}
	\zeta_p^{\frac{B^2}{4zm_p}\delta^2+ \bk{\frac{B}{m_p}-\rho}\delta} \\
	& =  \eta(-m_p) G_m G \sum_{z\in \mathbb{F}_p^*}
	\zeta_p^{ \frac{z}{m_p}}
	\Bk{
		\zeta_p^{-\frac{zm_p}{B^2} \bk{\frac{B}{m_p}-\rho}^2 } \eta(zm_p)G-1} \\
	& =  G_m G^2  \sum_{z\in \mathbb{F}_p^*}  \eta(-z)
	\zeta_p^{\frac{2B-\rho m_p}{B^2} \rho z }+ \eta(-m_p) G_m G \\
	& =  \left\{
	\begin{array}{lll}
	\eta(-m_p) G_m G	& &\textup{ if } \rho m_p = 2  B\\
	\eta(-1)G_m G^3\eta(2B\rho-m_p \rho^2)+ \eta(-m_p) G_m G    	& & \textup{ if }  \rho m_p \neq 2  B
	\end{array} \right.\\
	& =  \left\{
	\begin{array}{lll}
	\eta(-m_p) G_m G  	& &\textup{ if }  \rho m_p =2  B,\\
	\Bk{\eta(2B\rho-m_p \rho^2)p+ \eta(-m_p) }G_m G    	& & \textup{ if } \rho m_p \neq 2  B.
	\end{array} \right.
	\end{align*}
	Set $ \Delta = B^2-m_p A  $. If $ A \neq 0 $, then
	\begin{align*}
	\Omega_4
	& = G_m \sum_{z\in \mathbb{F}_p^*} \eta(z)
	\sum_{y\in \mathbb{F}_p^*}\zeta_p^{-\frac{m_p}{4z}y^2-y}
	\Bk{ \zeta_p^{  \frac{B^2}{4Az}y^2+\frac{\rho B}{A}y+\frac{\rho^2}{A}z } \eta\big(-\frac{A}{4z}\big) G-1 }\\
	& = \eta(-A)G_m G \sum_{z\in \mathbb{F}_p^*}
	\zeta_p^{\frac{ \rho^2}{A}z}
	\sum_{y\in \mathbb{F}_p^*}\zeta_p^{\frac{\Delta}{4Az}y^2+\bk{ \frac{\rho B}{A}-1}y}
	+ \eta(-m_p) G_m G,
	\end{align*}
	since $ G_m \sum_{z\in \mathbb{F}_p^*}  \eta(z)
	\sum_{y\in \mathbb{F}_p^*}\zeta_p^{-\frac{m_p}{4z}y^2-y} =- \eta(-m_p) G_m G$.
	
	Therefore, if $ A \neq 0 $ and $ \Delta =0 $, then $ \eta(-A)  =\eta(-m_p) $ and
	\begin{align*}
	\Omega_4
	& =\eta(-A) G_m G \sum_{z\in \mathbb{F}_p^*}
	\zeta_p^{\frac{ \rho^2}{A}z}
	\sum_{y\in \mathbb{F}_p^*}\zeta_p^{\bk{ \frac{\rho B}{A}-1}y}
	+ \eta(-m_p) G_m G\\
	& =  \left\{
	\begin{array}{lll}
	\Bk{-(p-1 )\eta(-A)   +\eta(-m_p)} G_m G	& &\textup{ if } \rho B = A\\
	\Bk{\eta(-A)   +\eta(-m_p)} G_m G   	& & \textup{ if }  \rho B \neq A
	\end{array} \right.\\
	& =  \left\{
	\begin{array}{lll}
	-(p-2 ) \eta(-m_p)  G_m G	& &\textup{ if } \rho B = A,\\
	2 \eta(-m_p)  G_m G   	& & \textup{ if }  \rho B \neq A.
	\end{array} \right.
	\end{align*}
	Suppose that $ A \neq 0 $ and $ \Delta \neq 0 $. We denote $ f(\rho)=-m_p \rho^2 +2B\rho-A $ as before. Then
	\begin{align*}
	\Omega_4
	& =\eta(-A) G_m G \sum_{z\in \mathbb{F}_p^*}
	\zeta_p^{\frac{ \rho^2}{A}z}
	\Bk{\zeta_p^{-\frac{Az}{\Delta}\bk{ \frac{\rho B}{A}-1}^2}\eta\bk{ \frac{\Delta}{4Az}}G  -1}
	+ \eta(-m_p) G_m G\\
	& = G_m G^2 \sum_{z\in \mathbb{F}_p^*} \eta(-z\Delta) \zeta_p^{f(\rho)\frac{z}{\Delta}}+\bk{\eta(-A)+\eta(-m_p)} G_m G\\
	& =  \left\{
	\begin{array}{lll}
	\bk{\eta(-A)+\eta(-m_p)} G_m G	& &\textup{ if } f(\rho) = 0,\\
	\Bk{\eta(f(\rho))p+\eta(-A)   +\eta(-m_p)} G_m G   	& & \textup{ if }  f(\rho ) \neq 0.
	\end{array} \right.
	\end{align*}
	This finishes the proof.
\end{proof}

The following lemmas will be required when calculating the frequency of each component in $ C_D $.
\begin{lemma}\label{lem:N(A,B)}
	For $ A \in \mathbb{F}_p^* $ and $ B \in \mathbb{F}_p $, define
	\begin{equation*}
	N(A,B) := \#\left\{x\in \mathbb{F}_{r}  :  \mathrm{Tr}(x^2)=A,\mathrm{Tr}(x)=B  \right\}.
	\end{equation*}
	Denote $ \Delta=B^2-m_p A $. Then the following assertions hold. \\
	$ (1) $ If $2 \mid m$ and $ m_p=0 $, then		
	\begin{align*}
	N(A,B) =  \left\{
	\begin{array}{lll}
	p^{m-2} - p^{-1}G_m  	\phantom{\eta(\Delta)}  & &\textup{ if }   B=0,\\
	p^{m-2}         	& & \textup{ if }  B\neq 0.
	\end{array} \right.
	\end{align*}
	$ (2) $ If $2 \mid m$ and $ m_p\neq 0 $, then
	\begin{align*}
	N(A,B) & =  \left\{
	\begin{array}{lll}
	p^{m-2}  	& &\textup{ if }  \Delta=0,\\
	p^{m-2} + \eta(\Delta)p^{-1}G_m   & & \textup{ if }   \Delta \neq 0.
	\end{array} \right.
	\end{align*}
	$ (3) $ If $2 \nmid m$ and $ m_p=0 $, then
	\begin{align*}
	N(A,B)
	=  \left\{
	\begin{array}{lll}
	p^{m-2}  + \eta(-A) p^{-1} G_m G      & &\textup{ if }  B=0,\\
	p^{m-2}  	& &\textup{ if }    B\neq 0.
	\end{array} \right.
	\end{align*}
	$ (4) $ If $2 \nmid m$ and $ m_p\neq 0 $, then
	\begin{align*}
	N(A,B) & =  \left\{
	\begin{array}{lll}
	p^{m-2}+\eta(-m_p)(p-1)p^{-2} G_m G 	& &\textup{ if }   \Delta=0,\\
	p^{m-2}-\eta(-m_p)p^{-2} G_m G    	& & \textup{ if }  \Delta\neq 0.
	\end{array} \right.
	\end{align*}
\end{lemma}
\begin{proof}
	By definition, we see that
	\begin{align}\label{eq:N(A,B)}
	N(A,B)
	&=\sum_{ x \in \mathbb{F}_r}
	\Big( \dfrac{1}{p}\sum_{y\in \mathbb{F}_p}\zeta_p^{y(\mathrm{Tr}(x )-B)}\Big)  \Big( \dfrac{1}{p}\sum_{z\in \mathbb{F}_p}\zeta_p^{z(\mathrm{Tr}(x^2)-A)}\Big)
	\nonumber \\
	&= p^{m-2}+p^{-2}(S_1+S_2+S_3 ),
	\end{align}
	where	
	\begin{align*}	
	S_1& =\sum_{x\in \mathbb{F}_r}
	\sum_{y\in \mathbb{F}_p^*}\zeta_p^{y(\mathrm{Tr}(x )-B)}
	= \sum_{y\in \mathbb{F}_p^*} \zeta_p^{-By}
	\sum_{x\in \mathbb{F}_r}
	\zeta_p^{y\mathrm{Tr}(x ) }
	=0,\\
	S_2&=\sum_{ x \in \mathbb{F}_r}
	\sum_{z\in \mathbb{F}_p^*}\zeta_p^{z(\mathrm{Tr}(x^2)-A)}  ,\\
	S_3&=\sum_{ x \in \mathbb{F}_r}
	\sum_{y\in \mathbb{F}_p^*}\zeta_p^{y(\mathrm{Tr}(x )-B)}
	\sum_{z\in \mathbb{F}_p^*}\zeta_p^{z(\mathrm{Tr}(x^2)-A)} .
	\end{align*}
	Note that $ A \neq 0 $. It is easily verified that
	\begin{align*}
	S_2 & =  \left\{
	\begin{array}{lll}
	- G_m 	& &\textup{ if }  2 \mid m ,\\
	\eta(-A)  G_m G 	& & \textup{ if } 2 \nmid m .
	\end{array} \right.
	\end{align*}
	To determine $ S_3 $, we observe that
	\begin{align*}
	S_3&=\sum_{y\in \mathbb{F}_p^*}\zeta_p^{-By}
	\sum_{z\in \mathbb{F}_p^*}\zeta_p^{ -Az}
	\sum_{ x \in \mathbb{F}_r}
	\zeta_p^{ \mathrm{Tr}(zx^2+yx )}  \\
	& = G_m \sum_{z\in \mathbb{F}_p^*}\eta_m(z)\zeta_p^{-Az}
	\sum_{y\in \mathbb{F}_p^*}\zeta_p^{  -\frac{m_p}{4z}y^2-By}    .
	\end{align*}
	
	Case (1): Suppose that $2 \mid m$ and $ m_p=0 $. Then
	\begin{align*}
	S_3& = G_m \sum_{z\in \mathbb{F}_p^*} \zeta_p^{-Az}
	\sum_{y\in \mathbb{F}_p^*}\zeta_p^{  -By} \\
	&=  \left\{
	\begin{array}{lll}
	-(p-1) G_m                & &\textup{ if } B=0,\\
	G_m 	& &\textup{ if }  B \neq 0 .
	\end{array} \right.
	\end{align*}
	
	Case (2): Suppose that $2 \mid m$ and $ m_p\neq 0 $. Then
	\begin{align*}
	S_3& = G_m \sum_{z\in \mathbb{F}_p^*} \zeta_p^{-Az}
	\sum_{y\in \mathbb{F}_p^*}\zeta_p^{ -\frac{m_p}{4z}y^2 -By} \\
	& =  G_m \sum_{z\in \mathbb{F}_p^*} \zeta_p^{-Az}
	\Bk{\zeta_p^{ \frac{B^2}{m_p}z }\eta\bk{ -\frac{m_p}{4z}}  G -1  } \\
	& =  G_m G \sum_{z\in \mathbb{F}_p^*}
	\zeta_p^{ \frac{ \Delta}{m_p} z } \eta( - m_p z)   +G_m\\
	&=  \left\{
	\begin{array}{lll}
	G_m                & &\textup{ if } \Delta=0,\\
	(\eta(\Delta)p+1)G_m 	& &\textup{ if }   \Delta \neq 0 .
	\end{array} \right.
	\end{align*}
	
	Case (3): Suppose that $2 \nmid m$ and $ m_p=0 $. Then
	\begin{align*}
	S_3 & = G_m \sum_{z\in \mathbb{F}_p^*}\eta(z)\zeta_p^{-Az}
	\sum_{y\in \mathbb{F}_p^*}\zeta_p^{  -By}  \\
	&=  \left\{
	\begin{array}{lll}
	(p-1) \eta (-A) G_m G                & &\textup{ if } B=0,\\
	-\eta (-A) G_m G  	& &\textup{ if }   B \neq 0 .
	\end{array} \right.
	\end{align*}
	
	Case (4): Suppose that $2 \nmid m$ and $ m_p \neq 0 $. Then
	\begin{align*}
	S_3   	& = G_m \sum_{z\in \mathbb{F}_p^*}\eta(z)\zeta_p^{-Az}
	\sum_{y\in \mathbb{F}_p^*}\zeta_p^{  -\frac{m_p}{4z}y^2-By}    \\
	& = G_m \sum_{z\in \mathbb{F}_p^*}\eta(z)\zeta_p^{-Az}
	\Bk{\zeta_p^{ \frac{B^2}{m_p}z }\eta\bk{ -\frac{m_p}{4z}}  G -1  }\\
	& = \eta( - m_p )G_m G \sum_{z\in \mathbb{F}_p^*}
	\zeta_p^{ \frac{ \Delta}{m_p} z } -\eta (-A) G_m G\\
	& =  \left\{
	\begin{array}{lll}
	\Bk{ (p-1) \eta( - m_p ) -\eta (-A) }G_m G    & &\textup{ if } \Delta=0,\\
	\Bk{ - \eta( - m_p ) -\eta (-A) }G_m G	& &\textup{ if }   \Delta \neq 0 .
	\end{array} \right.
	\end{align*}
	
	Combining the above with~\eqref{eq:N(A,B)} gives us the desired conclusion, which completes the whole proof.  	
\end{proof}

\begin{lemma}\label{lem:delta=1and-1}
	Suppose that $ A \in \mathbb{F}_p^* $, $ B \in \mathbb{F}_p  $ and $ \Delta=B^2-m_pA $.	For $ i\in\{1,-1\} $, let $ T_i $ denote the number of the pairs $ (A,B) $ such that $ \eta(\Delta)=i $. Then we have
	\begin{align*}
	\phantom{-}T_1&= \frac{1}{2}(p-1)(p-2),\\
	T_{-1}&= \frac{1 }{2}(p-1)p.
	\end{align*}
\end{lemma}
\begin{proof}
	We first consider that $ B=0 $. So $ \Delta=-m_p A $ and the number of the pairs $ (A,0) $ satisfying $ \eta(\Delta)=i $ is $ (p-1)/2 $.
	
	Now suppose that $ B \neq 0 $. Note that $ \Delta=B^2-m_pA $ yields that
	\begin{align*}
	\frac{m_pA}{\Delta} + 1 = \frac{B^2}{\Delta}.
	\end{align*}
	Set $ p=2h+1 $. We count the number of the pairs $ (A,B^2) $ for a fixed $ \Delta_0 $ such that $ \eta(\Delta_0)=1 $ (resp. $ \eta(\Delta_0)=-1 $). It follows from Lemma~\ref{lemN=2} that this number is equal to
	\begin{align*}
	(0,0)^{(2,p)}+ (1,0)^{(2,p)} =h-1~~ (\text{ resp. } (0,1)^{(2,p)}+ (1,1)^{(2,p)} =h) .
	\end{align*}
	So the number of the pairs $ (A,B) $ such that $ \eta(\Delta_0)=1 $ (resp. $ \eta(\Delta_0)=-1 $) is $ 2(h-1) $ (resp. $ 2h $). We conclude that $ T_1=(p-1)/2+(p-1)(h-1) $ (resp. $ T_{-1}=(p-1)/2+(p-1)h $), and hence the result follows.
\end{proof}
\begin{lemma}\label{lem:gamma}
	Suppose that $ A \in \mathbb{F}_p^* $, $ B \in \mathbb{F}_p  $ and $ \Delta=B^2-m_pA \neq 0 $.	For $ i\in\{1,-1\} $, let $ \gamma_i $ denote the number of the pairs $ (A,B) $ such that $ \eta(A)=i $. Then we have
	\begin{align*}
	\phantom{-}\gamma_1
	=\left\{
	\begin{array}{lll}
	\frac{1}{2}(p-1)(p-2)   &\textup{ if } \eta(m_p)=1,\\
	\frac{1 }{2}(p-1)p	 &\textup{ if }   \eta(m_p)=-1.
	\end{array} \right.
	\end{align*}
	and
	\begin{align*}
	\gamma_{-1}
	=\left\{
	\begin{array}{lll}
	\frac{1}{2}(p-1) p  & &\textup{ if } \eta(m_p)=1,\\
	\frac{1 }{2}(p-1)(p-2)	& &\textup{ if }   \eta(m_p)=-1.
	\end{array} \right.
	\end{align*}
\end{lemma}
\begin{proof}
	The proof is similar to that of Lemma~\ref{lem:delta=1and-1} and so it is omitted here.
\end{proof}
With the above preparations,
we are ready to determine the complete weight enumerator of $ C_D $, which is stated in the next theorem and then illustrated with some examples.
\begin{theorem}\label{thm:cweofcode}
	Let $ C_D $ be the linear code defined by~\eqref{def:CD}, where the defining set
	$
	D=\left\{ x\in \mathbb{F}_{r}  :  \mathrm{Tr}(x)=1, \mathrm{Tr}(x^2)=0 \right\}
	$. Define $f(x)=-m_px^2+2Bx-A \in \mathbb{F}_{p}[x] $ and $ \Delta= B^2-m_p A $. Assume that $ \rho \in \mathbb{F}_p^* $.  \\
	$ (1) $ If $2 \mid m$ and $ m_p=0 $, then $ C_D $ has parameters $ [	p^{m-2}, m] $. Its complete weight enumerator is given as follows.
	
	$ N_{\rho}=0 $. This value occurs only once.
	
	$ N_{\rho}=p^{m-3}$. This value occurs
	$  p^{m -1} -p $    times.
	
	$ N_{\rho}=p^{m-3} +(-1)^{\frac{(p-1)m}{4}} p^{\frac{m-4}{2}} $.
	This value occurs $  (p-1)p^{m-2}  $ times.
	
	$ N_{\rho}=\left\{\begin{array}{ll}
	p^{m-2}
	&\text{  if  } \rho=\rho_0\\
	0
	&\text{  if  } \rho \neq \rho_0
	\end{array}\right. $
	as $ \rho_0 $ runs through $ \mathbb{F}_p^* $.
	Each value occurs only once.
	
	$ N_{\rho}=\left\{\begin{array}{ll}
	p^{m-3} -(p-1)(-1)^{\frac{(p-1)m}{4}} p^{\frac{m-4}{2}}
	&\text{  if  } \rho=\rho_0\\
	p^{m-3} +(-1)^{\frac{(p-1)m}{4}} p^{\frac{m-4}{2}}
	&\text{  if  } \rho \neq \rho_0
	\end{array}\right. $
	as $ \rho_0 $ runs through $ \mathbb{F}_p^* $. Each value occurs $  (p-1)p^{m-2}  $ times. \\
	$ (2) $ If $2 \mid m$ and $ m_p\neq 0 $, then $ C_D $ has parameters $ [n, m] $, where  	
	\begin{equation*}
	n=	p^{m-2}-(-1)^{\frac{(p-1)m}{4}}p^{\frac{m-2}{2}}.
	\end{equation*}
	Its complete weight enumerator is given as follows.
	
	$  	N_{\rho}=0 $. This value occurs only once,
	
	$ N_{\rho}=p^{m-3}$. This value occurs
	$  p^{m -2} -1 $    times.
	
	$ N_{\rho}=p^{m-3} - (-1)^{\frac{(p-1)m}{4}}p^{\frac{m-4}{2}} $.
	This value occurs
	$ \frac{p-1}{2}\bk{ p^{m -1} + (-1)^{\frac{(p-1)m}{4}}p^{\frac{m}{2}} } $    times.
	
	$ N_{\rho}=\left\{\begin{array}{ll}
	n
	&\text{  if  } \rho=\rho_0\\
	0
	&\text{  if  } \rho \neq \rho_0
	\end{array}\right. $
	as $ \rho_0 $ runs through $ \mathbb{F}_p^* $. Each value occurs only once.
	
	$ N_{\rho}=\left\{\begin{array}{ll}
	p^{m-3}- (-1)^{\frac{(p-1)m}{4}}p^{\frac{m-2}{2}}
	&\text{  if  } \rho=\rho_0\\
	p^{m-3}
	&\text{  if  } \rho \neq \rho_0
	\end{array}\right. $
	as $ \rho_0 $ runs through $ \mathbb{F}_p^* $.
	Each value occurs $ p^{m-2}-1  $ times.	
	
	$ N_{\rho}=\left\{\begin{array}{ll}
	p^{m-3}-(p-1)(-1)^{\frac{(p-1)m}{4}}p^{\frac{m-4}{2}}
	&\text{  if  } \rho=\rho_0\\
	p^{m-3} + (-1)^{\frac{(p-1)m}{4}}p^{\frac{m-4}{2}}
	&\text{  if  } \rho \neq \rho_0
	\end{array}\right. $
	as $ \rho_0 $ runs through $ \mathbb{F}_p^* $.
	Each value occurs $ n $ times.
	
	$ N_{\rho}=\left\{\begin{array}{ll}
	p^{m-3}-(p-1)(-1)^{\frac{(p-1)m}{4}}p^{\frac{m-4}{2}}
	&\text{  if  }  \rho =\rho_0,\rho_1\\
	p^{m-3} + (-1)^{\frac{(p-1)m}{4}}p^{\frac{m-4}{2}}
	&\text{  if  }  \rho  \neq  \rho_0,\rho_1
	\end{array}\right. $
	as two distinct elements $ \rho_0,\rho_1 $ run through $ \mathbb{F}_p^* $.
	Each value occurs $ n $ times. \\
	$ (3) $ If $2 \nmid m $ and $ m_p=0 $, then $ C_D $ has parameters $ [	p^{m-2}, m] $. Its complete weight enumerator is given as follows.
	
	$ N_{\rho}=0 $. This value occurs only once.
	
	$ N_{\rho}=p^{m-3}$. This value occurs $  p^{m -1} -p $ times.
	
	$ N_{\rho}=p^{m-3}+\eta(\rho) (-1)^{\frac{(p-1)(m+1)}{4}} p^{\frac{m-3}{2}} $.
	This value occurs
	$ \frac{1}{2} (p-1)p^{m -2}  $    times.
	
	$ N_{\rho}=p^{m-3}-\eta(\rho) (-1)^{\frac{(p-1)(m+1)}{4}} p^{\frac{m-3}{2}} $.
	This value occurs
	$ \frac{1}{2} (p-1)p^{m -2}  $    times.
	
	$ N_{\rho}=\left\{\begin{array}{ll}
	p^{m-2}
	&\text{  if  } \rho=\rho_0\\
	0
	&\text{  if  } \rho \neq \rho_0
	\end{array}\right. $
	as $ \rho_0 $ runs through $ \mathbb{F}_p^* $.
	Each value occurs only once.
	
	$ N_{\rho}=\left\{\begin{array}{ll}
	p^{m-3}
	&\text{  if  } \rho=\rho_0\\
	p^{m-3} + \eta(\rho-\rho_0) (-1)^{\frac{(p-1)(m+1)}{4}} p^{\frac{m-3}{2}}
	&\text{  if  } \rho \neq \rho_0
	\end{array}\right. $
	as $ \rho_0 $ runs through $ \mathbb{F}_p^* $.
	Each value occurs $  \frac{1}{2}(p-1)p^{m-2}  $ times.
	
	$ N_{\rho}=\left\{\begin{array}{ll}
	p^{m-3}
	&\text{  if  } \rho=\rho_0\\
	p^{m-3} - \eta(\rho-\rho_0) (-1)^{\frac{(p-1)(m+1)}{4}} p^{\frac{m-3}{2}}
	&\text{  if  } \rho \neq \rho_0
	\end{array}\right. $
	as $ \rho_0 $ runs through $ \mathbb{F}_p^* $.
	Each value occurs $  \frac{1}{2}(p-1)p^{m-2}  $ times. 	\\
	$ (4) $ If $2 \nmid m $ and $ m_p\neq 0 $, then $ C_D $ has parameters $ [	n, m] $, where
	\begin{equation*}
	n=p^{m-2}- \eta(-m_p)(-1)^{\frac{(p-1)(m+1)}{4}} p^{\frac{m-3}{2}} .
	\end{equation*}
	Its complete weight enumerator is given as follows.
	
	$  	N_{\rho}=0 $. This value occurs only once.
	
	$ N_{\rho}=p^{m-3}$. This value occurs
	$  n+\eta(-m_p)(-1)^{\frac{(p-1)(m+1)}{4}} p^{\frac{m-1}{2}}-1 $ times.

	$ N_{\rho}=\left\{\begin{array}{ll}
	n
	&\text{  if  } \rho=\rho_0\\
	0
	&\text{  if  } \rho \neq \rho_0
	\end{array}\right. $
	as $ \rho_0 $ runs through $ \mathbb{F}_p^* $. Each value occurs only once.
	
	$
	N_{\rho}
	=\left\{\begin{array}{lr}
	p^{m-3}   &\textup{  if  }  \rho  = \rho_0\\
	p^{m-3}+   \eta(-m_p)\eta(\rho^2-\rho  \rho_0 )(-1)^{\frac{(p-1)(m+1)}{4}} p^{\frac{m-3}{2}}  &\textup{  if  }  \rho  \neq  \rho_0
	\end{array}
	\right.
	$
	as $ \rho_0 $ runs through $ \mathbb{F}_p^* $. Each value occurs $ n  $ times.
	
	$ N_{\rho} =\left\{\begin{array}{lr}
	p^{m-3}- \eta(-m_p)(-1)^{\frac{(p-1)(m+1)}{4}} p^{\frac{m-3}{2}}   &\textup{  if  }  \rho  = \rho_0\\
	p^{m-3}    &\textup{  if  }  \rho  \neq  \rho_0
	\end{array}
	\right.
	$
	as $ \rho_0 $ runs through $ \mathbb{F}_p^* $. Each value occurs
	$ n+\eta(-m_p)(-1)^{\frac{(p-1)(m+1)}{4}} p^{\frac{m-1}{2}} -1 $ times.
	
	$ N_{\rho }=\left\{\begin{array}{lr}
	p^{m-3}    &\textup{ if  } \rho=\rho_0,\rho_1\\
	p^{m-3}+ \eta\bk{-m_p (\rho-\rho_0) (\rho-\rho_1) } (-1)^{\frac{(p-1)(m+1)}{4}} p^{\frac{m-3}{2}}   &\textup{ if  } \rho \neq  \rho_0,\rho_1
	\end{array}
	\right.  $
	as two distinct elements $ \rho_0,\rho_1 $ run through $ \mathbb{F}_p^* $.	
	Each value occurs $ n $ times.
	
	$  N_{\rho}=
	p^{m-3} +   \eta(-m_p) \eta\bk{m_p^2  \rho^2-\Delta} (-1)^{\frac{(p-1)(m+1)}{4}} p^{\frac{m-3}{2}} $, as $ \Delta $ runs through $ \mathbb{F}_p^* $ such that $ \eta(\Delta)=-1 $. Each value occurs $ n $ times.
	
	$ N_{\rho}=\left\{\begin{array}{lr}
	p^{m-3} -   \eta(m_p) (-1)^{\frac{(p-1)(m+1)}{4}} p^{\frac{m-3}{2}}  &\textup{ if  } \rho=\rho_0\\
	p^{m-3} +   \eta(-m_p) \eta\bk{m_p^2 (\rho-\rho_0)^2-\Delta}
	(-1)^{\frac{(p-1)(m+1)}{4}} p^{\frac{m-3}{2}}  &\textup{ if  } \rho \neq  \rho_0
	\end{array}
	\right.
	$  as $ \rho_0 $ runs through $ \mathbb{F}_p^* $ and $ \Delta $ runs through $ \mathbb{F}_p^* $ such that $ \eta(\Delta)=-1 $. Each value occurs $ n $ times.
	
\end{theorem}
\begin{proof}	
	From the definition, this code has length $ n=\#D $ which follows from Lemma~\ref{lem:codelength} and dimension $ m $.
	As before $ A=\mathrm{Tr}(a^2) $, $ B=\mathrm{Tr}(a ) $ for $ a \in \mathbb{F}_r^* $. Recall that
	$ N_{\rho} = \frac{n}{p}+p^{-3}(\Omega_2+\Omega_3+\Omega_4) $ by~\eqref{eq:N_rho} for $ \rho \in \mathbb{F}_p^* $. We will divide the proof into four parts and employ Lemmas~\ref{lem:omega2},~\ref{lem:omega3} and~\ref{lem:omega4} to compute $ N_{\rho} $.

	$ (1) $ We first consider the case that $m$ is even and $ m_p = 0 $.  In this case the length is $ n=p^{m-2} $. If $ a \in \mathbb{F}_p^* $, then $ A = 0 $, $ B =  0 $, consequently
	\begin{align*}
	N_{\rho}
	&=\frac{n}{p}+p^{-3}(\Omega_2+\Omega_3+\Omega_4)\\
	&=\left\{\begin{array}{lr}
	p^{m-3} +p^{-3} \bk{ (p-1)r -(p-1)G_m + (p-1)G_m}    &\textup{  if  }  \rho=a,\\
	p^{m-3} +p^{-3}\bk{-r-(p-1)G_m + (p-1)G_m}   &\textup{  if  }  \rho \neq a.
	\end{array}
	\right.\\
	&=\left\{\begin{array}{lr}
	p^{m-2}   &\textup{  if  }  \rho=a,\\
	0 &\textup{  if  }  \rho \neq a.
	\end{array}
	\right.
	\end{align*}
	Each value occurs only once.

	Suppose that $a \in \mathbb{F}_r^*\backslash \mathbb{F}_p^* $. Under this assumption we have $ \Omega_2=0 $.
	If $ A=B=0 $, then
	\begin{align*}
	N_{\rho}&=\frac{n}{p}+p^{-3}( \Omega_3+\Omega_4)\\
	& =p^{m-3}+p^{-3}(-(p-1)G_m+(p-1)G_m)=p^{m-3}.
	\end{align*}
	By Lemma~\ref{lem:codelength}, the frequency is $ p^{m-2}+p^{-1}(p-1)G_m -p $ as $ a\notin \mathbb{F}_p  $.
	
	If $ A\neq 0 $,  $B=0 $, then
	\begin{align*}
	N_{\rho}&=\frac{n}{p}+p^{-3}( \Omega_3+\Omega_4)\\
	& =p^{m-3} +p^{-3} \bk{G_m-G_m} = p^{m-3} .
	\end{align*}
	By Lemma~\ref{lem:N(A,B)}, the frequency is $ (p-1)\bk{p^{m-2} - p^{-1}G_m } $ for all $ A\in \mathbb{F}_p^*$.
	
	Hence we conclude that $ N_{\rho} =p^{m-3} $ occurs $ p^{m-1}-p $ times.
	
	If $ A= 0 $,  $B\neq0 $, then
	\begin{align*}
	N_{\rho}&=\frac{n}{p}+p^{-3}( \Omega_3+\Omega_4)\\
	& =p^{m-3} +p^{-3}   \bk{ -(p-1) G_m  -G_m}=p^{m-3} -p^{-2}G_m  .
	\end{align*}
	It follows from Lemma~\ref{lem:codelength} that this value occurs $  (p-1)p^{m-2}  $ for all $ B \in \mathbb{F}_p^*$.
	
	If $ A B\neq0 $, then
	\begin{align*}
	N_{\rho}&=\frac{n}{p}+p^{-3}( \Omega_3+\Omega_4)\\
	& =\left\{\begin{array}{lr}
	p^{m-3} +p^{-3}  \bk{1 + (p^2-p-1)} G_m       &\textup{  if  }  A=2 \rho B,\\
	p^{m-3} +p^{-3}  \bk{1 - ( p+1) } G_m     &\textup{  if  }  A \neq 2 \rho B.
	\end{array}
	\right.\\
	& =\left\{\begin{array}{lr}
	p^{m-3} + p^{-2}  (p-1) G_m     &\textup{  if  }  \rho=\rho_0,\\
	p^{m-3} - p^{-2}   G_m    &\textup{  if  }  \rho\neq  \rho_0,
	\end{array}
	\right.
	\end{align*}
	where we denote $ \rho_0=\frac{A}{2B} \in \mathbb{F}_p^*$.
	By Lemma~\ref{lem:N(A,B)}, the frequency is
	$ (p-1) p^{m-2} 	$ for all $ B \in \mathbb{F}_p^*$.
	
	$ (2) $ We now consider the case that $m$ is even and $ m_p \neq 0 $. In this case the length is $ n=	p^{m-2}+p^{-1}G_m $. If $ a \in \mathbb{F}_p^* $, then $ A =a^2 m_p \neq 0 $, $ B =a  m_p \neq 0 $, consequently $ \Delta=B^2-m_p A=0 $ and
	\begin{align*}
	N_{\rho}
	&=\frac{n}{p}+p^{-3}(\Omega_2+\Omega_3+\Omega_4)\\
	&=\left\{\begin{array}{lr}
	p^{m-3}+p^{-2}G_m +p^{-3}\bk{ (p-1)r +G_m +(p^2-p-1)G_m}   &\textup{  if  }  \rho=a,\\
	p^{m-3}+p^{-2}G_m +p^{-3}\bk{ -r +G_m -(p+1)G_m}  &\textup{  if  }  \rho \neq a.
	\end{array}
	\right.\\
	&=\left\{\begin{array}{lr}
	p^{m-2}+p^{-1}G_m  &\textup{  if  }  \rho=a,\\
	0 &\textup{  if  }  \rho \neq a.
	\end{array}
	\right.
	\end{align*}
	Each value occurs only once.

	Suppose that $a \in \mathbb{F}_r^*\backslash \mathbb{F}_p^* $. Under this assumption we have $ \Omega_2=0 $.
	If $ A=B=0 $, then
	\begin{align*}
	N_{\rho}&=\frac{n}{p}+p^{-3}( \Omega_3+\Omega_4)\\
	&=p^{m-3}+p^{-2}G_m +p^{-3}\bk{ -(p-1)G_m -G_m}\\
	&=p^{m-3}.
	\end{align*}
	By Lemma~\ref{lem:codelength}, the frequency is $ p^{m-2}-1 $ as $ a \neq 0 $.
	
	If $ A=0 $, $ B \neq 0 $, then
	\begin{align*}
	N_{\rho}&=\frac{n}{p}+p^{-3}( \Omega_3+\Omega_4)\\
	&=\left\{\begin{array}{lr}
	p^{m-3}+p^{-2}G_m +p^{-3}\bk{ -(p-1)G_m + (p^2-p-1) G_m}  &\textup{  if  }  \rho m_p=2B,\\
	p^{m-3}+p^{-2}G_m +p^{-3}\bk{ -(p-1)G_m - ( p+1) G_m} &\textup{  if  }  \rho m_p\neq 2B.
	\end{array}
	\right.\\
	&=\left\{\begin{array}{lr}
	p^{m-3}+p^{-2}(p-1)G_m   &\textup{  if  }  \rho  = \rho_0,\\
	p^{m-3}-p^{-2}G_m   &\textup{  if  }  \rho  \neq  \rho_0,
	\end{array}
	\right.
	\end{align*}
	where we denote $ \rho_0=\frac{2B}{m_p} \in \mathbb{F}_p^* $. By Lemma~\ref{lem:codelength}, the frequency is $ p^{m-2}+p^{-1} G_m $.
	
	If $ A\neq 0 $ and $ \Delta =0 $, then
	\begin{align*}
	N_{\rho}&=\frac{n}{p}+p^{-3}( \Omega_3+\Omega_4)\\
	&=\left\{\begin{array}{lr}
	p^{m-3}+p^{-2}G_m +p^{-3}\bk{ G_m + (p^2-p-1) G_m}  &\textup{  if  } \rho B=A,\\
	p^{m-3}+p^{-2}G_m +p^{-3}\bk{ G_m - ( p+1) G_m} &\textup{  if  }  \rho B\neq A.
	\end{array}
	\right.\\
	&=\left\{\begin{array}{lr}
	p^{m-3}+p^{-1}G_m   &\textup{  if  }  \rho  = \rho_0,\\
	p^{m-3}    &\textup{  if  }  \rho  \neq  \rho_0,
	\end{array}
	\right.
	\end{align*}
	where we denote $ \rho_0=\frac{B}{m_p} \in \mathbb{F}_p^* $ since $ \Delta =0 $ and $ \rho B=A $ imply that $ \rho m_p=B $. By Lemma~\ref{lem:N(A,B)}, the frequency is $ p^{m-2}-1 $ as $ a \notin \mathbb{F}_p^* $.

	If $ A\neq 0 $ and $ \eta(\Delta) =1 $, then
	\begin{align*}
	N_{\rho}&=\frac{n}{p}+p^{-3}( \Omega_3+\Omega_4)\\
	&=\left\{\begin{array}{lr}
	p^{m-3}+p^{-2}G_m +p^{-3}\bk{ G_m + (p^2-2p-1) G_m}  &\textup{  if  } f(\rho)=0,\\
	p^{m-3}+p^{-2}G_m +p^{-3}\bk{ G_m - ( 2p+1) G_m} &\textup{  if  }  f(\rho)\neq 0.
	\end{array}
	\right.\\
	&=\left\{\begin{array}{lr}
	p^{m-3}+p^{-2}(p-1)G_m   &\textup{ if  }  \rho  = \rho_0,\rho_1,\\
	p^{m-3}-p^{-2} G_m     &\textup{ if  }  \rho  \neq  \rho_0,\rho_1,
	\end{array}
	\right.
	\end{align*}
	where we denote $ \rho_0,\rho_1 $ as two distinct roots of the equation $ f(\rho)=0  $ since $ \eta(\Delta) =1 $.
	By Lemma~\ref{lem:N(A,B)}, the frequency is $ p^{m-2}+p^{-1}G_m $.
	
	If $ A\neq 0 $ and $ \eta(\Delta) =-1 $, then
	\begin{align*}
	N_{\rho}&=\frac{n}{p}+p^{-3}( \Omega_3+\Omega_4)\\
	&=
	p^{m-3}+p^{-2}G_m +p^{-3}\bk{ G_m - G_m}  \\
	&=p^{m-3}+p^{-2}G_m .
	\end{align*}
	By Lemmas~\ref{lem:N(A,B)} and ~\ref{lem:delta=1and-1}, the frequency is $ \frac{1}{2}(p-1)\bk{p^{m-1}+G_m} $.
	
	$ (3) $ Suppose that $m$ is odd and $ m_p=0 $. In this case the length is $ n=p^{m-2} $. If $ a \in \mathbb{F}_p^* $, then $ A = 0 $, $ B =  0 $, consequently
	\begin{align*}
	N_{\rho}
	&=\frac{n}{p}+p^{-3}(\Omega_2+\Omega_3+\Omega_4)\\
	&=\left\{\begin{array}{lr}
	p^{m-3} +p^{-3}  (p-1)r      &\textup{  if  }  \rho=a,\\
	p^{m-3} +p^{-3}( -r )   &\textup{  if  }  \rho \neq a.
	\end{array}
	\right.\\
	&=\left\{\begin{array}{lr}
	p^{m-2}   &\textup{  if  }  \rho=a,\\
	0 &\textup{  if  }  \rho \neq a.
	\end{array}
	\right.
	\end{align*}
	Each value occurs only once.

	Suppose that $a \in \mathbb{F}_r^*\backslash \mathbb{F}_p^* $. Under this assumption we have $ \Omega_2=0 $.
	If $ A=B=0 $, then $ \Omega_3=\Omega_4 =0$ and
	\begin{align*}
	N_{\rho}=\frac{n}{p}+p^{-3}( \Omega_3+\Omega_4)=p^{m-3}.
	\end{align*}
	By Lemma~\ref{lem:codelength}, the frequency is $ p^{m-2}-p $ as $ a\notin \mathbb{F}_p  $.
	
	If $ A\neq 0 $,  $B=0 $, then
	\begin{align*}
	N_{\rho}&=\frac{n}{p}+p^{-3}( \Omega_3+\Omega_4)\\
	& =p^{m-3} +p^{-3} \bk{-\eta(-A) + \eta(-A)}G_m G= p^{m-3} .
	\end{align*}
	By Lemma~\ref{lem:N(A,B)}, the frequency is $ (p-1)p^{m-2}  $ for all $ A\in \mathbb{F}_p^*$.
	
	Hence we conclude that $ N_{\rho} =p^{m-3} $ occurs $ p^{m-2}-p + (p-1)p^{m-2} =p^{m-1}-p $ times.
	
	If $ A= 0 $,  $B\neq0 $, then
	\begin{align*}
	N_{\rho}&=\frac{n}{p}+p^{-3}( \Omega_3+\Omega_4)\\
	& =p^{m-3} +p^{-2}  \eta\bk{\frac{B\rho}{2}} G_m G.
	\end{align*}
	This indicates that $ N_{\rho} =p^{m-3} +p^{-2}  \eta (\rho) G_m G $ or $ N_{\rho} =p^{m-3} -p^{-2}  \eta (\rho) G_m G $.
	According to Lemma~\ref{lem:codelength}, the frequency of each value is $ \frac{1}{2}(p-1)p^{m-2}  $ for all $ B \in \mathbb{F}_p^*$.
	
	If $ A B\neq0 $, then
	\begin{align*}
	N_{\rho}&=\frac{n}{p}+p^{-3}( \Omega_3+\Omega_4)\\
	& =\left\{\begin{array}{lr}
	p^{m-3} +p^{-3}  \bk{-\eta(-A) + \eta(-A)} G_m G      &\textup{  if  }  A=2 \rho B,\\
	p^{m-3} +p^{-3}  \bk{-\eta(-A)  +\eta (2\rho B-A)p +\eta(-A)} G_m G   &\textup{  if  }  A \neq 2 \rho B.
	\end{array}
	\right.\\
	& =\left\{\begin{array}{lr}
	p^{m-3}     &\textup{  if  }  \rho=\rho_0,\\
	p^{m-3} +p^{-2}  \eta (2 B) \eta(\rho-\rho_0)  G_m G   &\textup{  if  }  \rho\neq  \rho_0,
	\end{array}
	\right.
	\end{align*}
	where we denote $ \rho_0=\frac{A}{2B} \in \mathbb{F}_p^*$. This induces that
	\begin{align*}
	N_{\rho}=\left\{\begin{array}{lr}
	p^{m-3}     &\textup{  if  }  \rho=\rho_0,\\
	p^{m-3} +p^{-2}  \eta(\rho-\rho_0)  G_m G   &\textup{  if  }  \rho\neq  \rho_0,
	\end{array}
	\right.
	\end{align*}
	or
	\begin{align*}
	N_{\rho}=\left\{\begin{array}{lr}
	p^{m-3}     &\textup{  if  }  \rho=\rho_0,\\
	p^{m-3} - p^{-2}  \eta(\rho-\rho_0)  G_m G   &\textup{  if  }  \rho\neq  \rho_0,
	\end{array}
	\right.
	\end{align*}
	According to Lemma~\ref{lem:N(A,B)}, each value occurs $ \frac{1}{2}(p-1)p^{m-2}  $ for all $ B \in \mathbb{F}_p^*$.
	
	$ (4) $ Assume that $m$ is odd and $ m_p \neq 0 $. In this case the length is $ 	n=p^{m-2}- p^{-2}\eta(-m_p)G_m G $. If $ a \in \mathbb{F}_p^* $, then $ A =a^2 m_p \neq 0 $, $ B =a  m_p \neq 0 $, consequently $ \Delta=B^2-m_p A=0 $,
	$ \eta(-A)=\eta(-m_p) $. Hence
	\begin{align*}
	N_{\rho}
	&=\frac{n}{p}+p^{-3}(\Omega_2+\Omega_3+\Omega_4)\\
	&=\left\{\begin{array}{lr}
	p^{-1}n  +p^{-3}\bk{ (p-1)r -(p-1)\eta(-m_p)G_mG}   &\textup{  if  }  \rho=a,\\
	p^{-1}n  +p^{-3}\bk{ -r +\eta(-m_p)G_mG}  &\textup{  if  }  \rho \neq a.
	\end{array}
	\right.\\
	&=\left\{\begin{array}{lr}
	n  &\textup{  if  }  \rho=a,\\
	0 &\textup{  if  }  \rho \neq a.
	\end{array}
	\right.
	\end{align*}
	Each value occurs only once.

	Suppose that $a \in \mathbb{F}_r^*\backslash \mathbb{F}_p^* $. Under this assumption we have $ \Omega_2=0 $.
	If $ A=B=0 $, then
	\begin{align*}
	N_{\rho}&=\frac{n}{p}+p^{-3}( \Omega_3+\Omega_4)\\
	&=p^{m-3}- p^{-3}\eta(-m_p)G_m G +p^{-3}  \eta(-m_p)G_m G\\
	&=p^{m-3}.
	\end{align*}
	By Lemma~\ref{lem:codelength}, the frequency is $ p^{m-2}+p^{-2}\eta(-m_p)(p-1)G_m G-1 $ as $ a \neq 0 $.
	
	If $ A=0 $, $ B \neq 0 $, then
	\begin{align*}
	N_{\rho}&=\frac{n}{p}+p^{-3}( \Omega_3+\Omega_4)\\
	&=\left\{\begin{array}{lr}
	p^{m-3}  &\textup{  if  }  \rho m_p=2B,\\
	p^{m-3} +p^{-2}  \eta(2\rho B-\rho^2 m_p)G_m G &\textup{  if  }  \rho m_p\neq 2B.
	\end{array}
	\right.\\
	&=\left\{\begin{array}{lr}
	p^{m-3}   &\textup{  if  }  \rho  = \rho_0,\\
	p^{m-3}+p^{-2}  \eta(-m_p)\eta(\rho^2- \rho  \rho_0)G_m G  &\textup{  if  }  \rho  \neq  \rho_0,
	\end{array}
	\right.
	\end{align*}
	where we denote $ \rho_0=\frac{2B}{m_p} \in \mathbb{F}_p^* $. By Lemma~\ref{lem:codelength}, the frequency is $n$.
	
	If $ A\neq 0 $ and $ \Delta =0 $, then $ \eta(-A)=\eta(-m_p) $, consequently
	\begin{align*}
	N_{\rho}&=\frac{n}{p}+p^{-3}( \Omega_3+\Omega_4)\\
	&=\left\{\begin{array}{lr}
	p^{-1} n +p^{-3}\bk{ -1- ( p-2) }\eta(-m_p)G_m G &\textup{  if  } \rho B=A,\\
	p^{-1} n +p^{-3}\bk{ -1 +2 }\eta(-m_p)G_m G&\textup{  if  }  \rho B\neq A.
	\end{array}
	\right.\\
	&=\left\{\begin{array}{lr}
	p^{m-3}-p^{-2} \eta(-m_p) G_m G   &\textup{  if  }  \rho  = \rho_0,\\
	p^{m-3}    &\textup{  if  }  \rho  \neq  \rho_0,
	\end{array}
	\right.
	\end{align*}
	where we denote $ \rho_0=\frac{B}{m_p} \in \mathbb{F}_p^* $ since $ \Delta =0 $ and $ \rho B=A $ imply that $ \rho m_p=B $. By Lemma~\ref{lem:N(A,B)}, the frequency is $ p^{m-2}+\eta(-m_p)(p-1)p^{-2} G_m G -1 $ as $ a \notin \mathbb{F}_p^* $.

	If $ A\neq 0 $ and $ \Delta \neq 0 $, then
	\begin{align*}
	N_{\rho}&=\frac{n}{p}+p^{-3}( \Omega_3+\Omega_4)\\
	&=\left\{\begin{array}{lr}
	p^{-1} n +p^{-3} \eta(-m_p) G_m G  &\textup{  if  } f(\rho)=0,\\
	p^{-1} n +p^{-3}\bk{ p\eta(f(\rho)) +\eta(-m_p)  } G_m G &\textup{  if  }  f(\rho)\neq 0.
	\end{array}
	\right.\\
	&=\left\{\begin{array}{lr}
	p^{m-3}    &\textup{ if  }  f(\rho)=0,\\
	p^{m-3}+p^{-2} \eta(f(\rho)) G_m G    &\textup{ if  } f(\rho)\neq 0,
	\end{array}
	\right.
	\end{align*}
	where $   f(\rho)=-m_p \rho^2+2B\rho-A  $. If $ \eta(\Delta) =1 $,
	then the equation $ f(\rho)=0  $ must have two distinct roots, which are denoted by $ \rho_0$ and $\rho_1 $. Thus $ f(\rho)$ can be represented as $ f(\rho)=-m_p(\rho-\rho_0)(\rho-\rho_1 )$. Therefore, if $ A\neq 0 $ and $ \eta(\Delta) =1 $, we have
	\begin{align*}
	N_{\rho }=\left\{\begin{array}{lr}
	p^{m-3}    &\textup{ if  } \rho=\rho_0,\rho_1,\\
	p^{m-3}+p^{-2} \eta(-m_p)\eta(\rho-\rho_0)\eta(\rho-\rho_1) G_m G    &\textup{ if  } \rho \neq  \rho_0,\rho_1.
	\end{array}
	\right.
	\end{align*}
	By Lemmas~\ref{lem:N(A,B)}, the frequency is $ n $. Moreover there are $ \frac{1}{2}(p-1)(p-2) $ such values by Lemma~\ref{lem:delta=1and-1}.
	
	If $ A\neq 0 $ and $ \eta(\Delta) =-1 $, then
	\begin{align*}
	N_{\rho}= p^{m-3}+p^{-2} \eta(f(\rho)) G_m G  .
	\end{align*}
	More precisely, by writing $ f(\rho)=-m_p\bk{\rho-  B /m_p }^2+  \Delta / m_p $, we deduce that
	\begin{align*}
	N_{\rho}=\left\{\begin{array}{lr}
	p^{m-3} - p^{-2} \eta(m_p) G_m G  &\textup{ if  } \rho=\rho_0, \\
	p^{m-3} + p^{-2} \eta(-m_p) \eta\bk{m_p^2 (\rho-\rho_0)^2-\Delta}G_m G   &\textup{ if  } \rho \neq  \rho_0 ,
	\end{array}
	\right.
	\end{align*}
	where we denote $ \rho_0=  B /m_p$ and $ B \neq 0 $. The number of such values is
	$ \frac{1}{2}(p-1)^2 $.
	On the other hand, if $ B= 0 $, then
	\begin{align*}
	N_{\rho}=
	p^{m-3} + p^{-2} \eta(-m_p) \eta\bk{m_p^2  \rho^2-\Delta}G_m G.
	\end{align*}
	The number of such values is
	$ \frac{1}{2}(p-1)  $.
	Again from Lemma~\ref{lem:N(A,B)}, each value occurs $ n $ times.
	
	
	This completes the whole proof.
\end{proof}

\begin{remark} For a fixed element $  b $ in $ \mathbb{F}_{p}^* $, if we define the set
	\begin{equation*}
	D_b=\left\{ x \in \mathbb{F}_{r} :  \mathrm{Tr}(x )=b, \mathrm{Tr}(x^2)=0 \right\},
	\end{equation*}	
	then we get the code $ C_{D_b} $ of the form~\eqref{def:CD}. It is interesting to see that any code $ C_{D_b} $ has the same codewords. Actually, there exists a mapping $  \phi_b $ such that
	\begin{align*}
	\phi_b :~ &D_1 \longrightarrow D_b\\
	&x~~\longmapsto bx.
	\end{align*}
	This implies that the code $ C_{D_b} $ is equal to $C_{D_1} $. Thus Theorem~\ref{thm:cweofcode} actually demonstrates the complete weight enumerator of
	$ C_{D_b} $ for all $ b\in \mathbb{F}_{p}^* $.
\end{remark}

The next result describes the weight distribution of $ C_D $.
\begin{corollary}\label{coro:weofcode}
	Let $ C_D $ be the linear code defined by~\eqref{def:CD}, where the defining set
	$
	D=\left\{ x\in \mathbb{F}_{r}  :  \mathrm{Tr}(x)=1, \mathrm{Tr}(x^2)=0 \right\}
	$. The following assertions hold.\\
	$ (1) $ If $ 2 \mid m $ and $ m_p=0 $, then the weight distribution of $ C_D $ is given in Table~\ref{wt:case1}.\\
	$ (2) $ If $2 \mid m$ and $ m_p\neq 0 $, then the weight distribution of $ C_D $ is given in Table~\ref{wt:case2}, where $ n=	p^{m-2}-(-1)^{\frac{(p-1)m}{4}}p^{\frac{m-2}{2}} $.\\
	$ (3) $ If $ 2 \nmid m $ and $ m_p=0 $, then the weight distribution of $ C_D $ is given in Table~\ref{wt:case3}.\\
	$ (4) $ If $2 \nmid m$ and $ m_p\neq 0 $, then the weight distribution of $ C_D $ is given in Table~\ref{wt:case4}, where
	$ n=p^{m-2}- \eta(-m_p)(-1)^{\frac{(p-1)(m+1)}{4}} p^{\frac{m-3}{2}}  $,
	\begin{align*}
	f_4=\left\{\begin{array}{lll}
	\frac{1}{2}(p-1) (p-2) n   & &\textup{ if  }  \eta(m_p)=1, \\
	\frac{1}{2}(p-1) p n       & &\textup{ if  }  \eta(m_p)=-1 ,
	\end{array}
	\right.
	\end{align*}
	and
	\begin{align*}
	f_5=\left\{\begin{array}{lll}
	\frac{1}{2}(p-1) p n      & &\textup{ if  }  \eta(m_p)=1, \\
	\frac{1}{2}(p-1) (p-2) n  & &\textup{ if  } \eta(m_p)=-1 .
	\end{array}
	\right.
	\end{align*}
	\begin{table}[htbp]
		\tabcolsep 2mm \caption{The weight distribution of $C_{D}$
			for $ 2 \mid m $ and $ m_p=0 $}\label{wt:case1}
		\begin{center}\begin{tabular}{ll}
				\hline\noalign{\smallskip}
				Weight    & Frequency             \\
				\noalign{\smallskip}\hline\noalign{\smallskip}
				0         &  1        \\
				$ p^{m-2}$
				&  $p-1 $\\
				$(p-1) p^{m-3} $
				&  $p^{m-1}-p  $ \\
				$ (p-1)\bk{p^{m-3} +(-1)^{\frac{(p-1)m}{4}} p^{\frac{m-4}{2}}}$
				& $   (p -1)p^{m-2} $\\
				$ (p-1)p^{m-3} -(-1)^{\frac{(p-1)m}{4}} p^{\frac{m-4}{2}} $
				& $   (p -1)^2 p^{m-2} $\\
				\noalign{\smallskip}\hline
			\end{tabular}
		\end{center}
	\end{table} 	
	\begin{table}[htbp]
		\tabcolsep 2mm \caption{The weight distribution of $C_{D}$
			for $ 2 \mid m $ and $ m_p\neq 0 $}\label{wt:case2}
		\begin{center}\begin{tabular}{ll}
				\hline\noalign{\smallskip}
				Weight    & Frequency              \\
				\noalign{\smallskip}\hline\noalign{\smallskip}
				0         &  1        \\
				$(p-1)p^{m-3} $
				&  $ p^{m-2}-1$  \\
				$(p-1)\bk{p^{m-3}- (-1)^{\frac{(p-1)m}{4}}p^{\frac{m-4}{2}} } $
				&  	
				$ \frac{p-1}{2}\bk{ p^{m -1} + (-1)^{\frac{(p-1)m}{4}}p^{\frac{m}{2}} } $   \\
				$ n $
				& $ p-1 $\\
				$ n-	p^{m-3} $
				& $ (p-1)(p^{m-2}-1) $  \\
				$(p-1)	p^{m-3}- (-1)^{\frac{(p-1)m}{4}}p^{\frac{m-4}{2}} $
				& $ (p-1)n $  		   \\
				$(p-1)	p^{m-3} - (p+1)(-1)^{\frac{(p-1)m}{4}}p^{\frac{m-4}{2}} $
				& $ \frac{1}{2}(p-1)(p-2)n $  		   \\		
				\noalign{\smallskip}\hline
			\end{tabular}
		\end{center}
	\end{table} 	 	
	\begin{table}[htbp]
		\tabcolsep 2mm \caption{The weight distribution of $C_{D}$
			for $ 2 \nmid m $ and $ m_p=0 $}\label{wt:case3}
		\begin{center}\begin{tabular}{ll}
				\hline\noalign{\smallskip}
				Weight    & Frequency              \\
				\noalign{\smallskip}\hline\noalign{\smallskip}
				0         &  1        \\
				$  p^{m -2} $	&  $ p-1 $ \\
				$(p-1) p^{m-3}  $
				& $ 2p^{m-1}-p^{m-2}-p$     \\
				$(p-1) p^{m-3}  -p^{\frac{m-3}{2}} $
				& $ \frac{1}{2} (p-1)^2 p^{m -2} $\\
				$(p-1) p^{m-3}  + p^{\frac{m-3}{2}} $
				& $ \frac{1}{2} (p-1)^2 p^{m -2} $   \\		
				\noalign{\smallskip}\hline
			\end{tabular}
		\end{center}
	\end{table}
	\begin{table}[htbp]
		\tabcolsep 2mm \caption{The weight distribution of $C_{D}$
			for $ 2 \nmid m $ and $ m_p \neq 0 $}\label{wt:case4}
		\begin{center}\begin{tabular}{ll}
				\hline\noalign{\smallskip}
				Weight    & Frequency              \\
				\noalign{\smallskip}\hline\noalign{\smallskip}
				0         & ~ 1        \\
				$(p-1) p^{m-3}  $
				&   $ n+\eta(-m_p)(-1)^{\frac{(p-1)(m+1)}{4}} p^{\frac{m-1}{2}}-1  $  \\
				$ n$
				&  $p-1$     \\
				$ n-p^{m-3}$
				& $ (p-1)\bk{ 2n+\eta(-m_p)(-1)^{\frac{(p-1)(m+1)}{4}} p^{\frac{m-1}{2}}-1 } $\\
				$n-p^{m-3}- (-1)^{\frac{(p-1)(m-1)}{4}} p^{\frac{m-3}{2}}$
				& $ f_4 $ \\
				$n-p^{m-3}+ (-1)^{\frac{(p-1)(m-1)}{4}} p^{\frac{m-3}{2}}$
				& $ f_5 $     \\		
				\noalign{\smallskip}\hline
			\end{tabular}
		\end{center}
	\end{table}
\end{corollary}	
\begin{proof}
	
	To calculate the weight distribution of $ C_D $, we will consider four distinct cases:
	\begin{enumerate}
		\item[(1)]  $2 \mid m$ and $ m_p=0 $,
		\item[(2)]  $2 \mid m$ and $ m_p \neq 0 $,
		\item[(3)]  $2 \nmid m$ and $ m_p=0 $,
		\item[(4)]  $2 \nmid m$ and $ m_p\neq 0 $.
	\end{enumerate}
	The results for cases $ (1) $, $ (2) $ and $ (3) $ will come from the corresponding complete weight enumerator as shown in Theorem~\ref{thm:cweofcode}, by observing that
	\begin{align*}
	\sum_{\rho \in \mathbb{F}_p^*} \eta(\rho-\rho_0) =\sum_{\rho \in \mathbb{F}_p} \eta(\rho-\rho_0)-\eta(-\rho_0)=-\eta(-\rho_0),
	\end{align*}
	where $ \rho_0 \in \mathbb{F}_p^* $. For the last case, we will present a direct calculation by considering the number of components $ \textup{Tr}(ax) $ of
	$ \mathsf{c}(a  ) $ that are equal to $ 0 $, which is
	denoted by $  N_{0} :=N_{0}(a) $, where $ a \in\mathbb{F}_r^* $. That is,
	\begin{align*}
	N_{0}
	=\#\{x\in \mathbb{F}_r :\mathrm{Tr}(x )=1,\mathrm{Tr}(x^2)=0,
	\mathrm{Tr}(ax)=0\}.
	\end{align*}
	Then the weight of codeword $ \mathsf{c}(a  ) $ is given by
	\begin{align}\label{eq:wt}
	wt(\mathsf{c}(a  ) ):= n-N_{0} ,
	\end{align}
	where $ n $ is the length of $ C_D $. Substituting $ \rho=0 $ in~\eqref{eq:N_rho} yields that
	\begin{align}\label{eq:N_0}
	N_{0}  = \dfrac{n}{p}+p^{-3}(\Omega_1'+\Omega_2'+\Omega_3'+\Omega_4'),
	\end{align}
	where 	
	\begin{align*}	
	\Omega_1'&=\sum_{ x \in \mathbb{F}_r}
	\sum_{\delta \in \mathbb{F}_p^*}
	\zeta_p^{\delta \mathrm{Tr}(ax) }
	=\sum_{\delta \in \mathbb{F}_p^*}
	\sum_{ x \in \mathbb{F}_r}\zeta_p^{\mathrm{Tr}(a\delta x)}=0,\\
	\Omega_2'&=\sum_{x\in \mathbb{F}_r}
	\sum_{y\in \mathbb{F}_p^*}\zeta_p^{y(\mathrm{Tr}(x )-1)}
	\sum_{\delta \in \mathbb{F}_p^*} \zeta_p^{\delta \mathrm{Tr}(ax) },\\
	\Omega_3'&=\sum_{ x \in \mathbb{F}_r}
	\sum_{z\in \mathbb{F}_p^*}\zeta_p^{z\mathrm{Tr}(x^2)}
	\sum_{\delta \in \mathbb{F}_p^*} \zeta_p^{\delta \mathrm{Tr}(ax) },\\
	\Omega_4'&=\sum_{ x \in \mathbb{F}_r}
	\sum_{y\in \mathbb{F}_p^*}\zeta_p^{y(\mathrm{Tr}(x )-1)}
	\sum_{z\in \mathbb{F}_p^*}\zeta_p^{z\mathrm{Tr}(x^2)}
	\sum_{\delta \in \mathbb{F}_p^*} \zeta_p^{\delta \mathrm{Tr}(ax)}.
	\end{align*}
	Note that $2 \nmid m$ and $ m_p\neq 0 $. In the same manner as in Lemmas~\ref{lem:omega2},~\ref{lem:omega3} and~\ref{lem:omega4}, we can show
	\begin{align*}
	\Omega_2'&
	=\left\{
	\begin{array}{lll}
	-r	& &\textup{ if } a \in\mathbb{F}_p^*, \\
	0   & & \textup{ otherwise.}
	\end{array} \right. \\
	\Omega_3'&
	=\left\{
	\begin{array}{lll}
	0   & & \textup{ if }     A = 0, \\
	(p-1)\eta(-A)G_mG   & & \textup{ if }    A \neq 0.
	\end{array} \right.\\
	\Omega_4 '  &
	=  \left\{
	\begin{array}{lll}
	-(p-1 ) \eta(-m_p) G_m G  	& &\textup{ if }  A=0,   B = 0,\\
	\eta(-m_p)  G_m G    	& & \textup{ if }  A=0,  B \neq 0,\\
	-(p-2)\eta(-m_p) G_m G	& &\textup{ if } A \neq 0,\Delta = 0 ,\\
	\bk{\eta(-A)+\eta(-m_p)} G_m G 	& &\textup{ if } A \neq 0,\Delta\neq 0 .
	\end{array} \right.
	\end{align*}
	The desired conclusion then follows from~\eqref{eq:wt},~\eqref{eq:N_0}, Lemmas~\ref{lem:codelength},~\ref{lem:N(A,B)} and~\ref{lem:gamma}.
\end{proof}

\begin{corollary}\label{coro:MDS}
	Let $ C_D $ be the linear code defined by~\eqref{def:CD}, where the defining set
	$
	D=\left\{ x\in \mathbb{F}_{r}  :  \mathrm{Tr}(x)=1, \mathrm{Tr}(x^2)=0 \right\}
	$. Then it is optimal with respect
	to the Griesmer bound only if $ m=3 $. Furthermore, when $ m=3 $, it is MDS and it has parameters $ [3,3,1] $ if $ p=3 $, or if $ p>3 $, it has parameters $ [p+1,3,p-1] $ if $ \eta(-3)=-1 $ and $ [p-1,3,p-3]  $ if $ \eta(-3)=1 $.	
\end{corollary}
\begin{proof}
	We only show the case of odd $ m $ since other case can be similarly verified.
	If $ m_p=0 $, then it follows from Corollary~\ref{coro:weofcode} that
	$ C_D $ has parameters $ [p^{m-2},m, (p-1) p^{m-3}  -p^{\frac{m-3}{2}}] $. Taking $ m=2m'+1 $, we can deduce that
	\begin{align*}
	\sum_{i=0}^{m-1} \ceil{\frac{d}{p^i}} 
	& = p^{2m'-1}-p^{m'-1}+1-\frac{p^{m'-1}-1}{p-1}.
	\end{align*}
	So the equation $  p^{2m'-1}-p^{m'-1}+1-\frac{p^{m'-1}-1}{p-1}=  p^{2m'-1} $ gives that
	$ m'=1 $, which means that $ m=3 $. As $ m_p=0 $, we must have $ p=3 $. Therefore if $ p=m=3 $, the code $ C_D $ is MDS with parameters $ [3,3,1] $.
	
	Suppose that $ m_p \neq 0 $ and $ \eta(-m_p)(-1)^{\frac{(p-1)(m+1)}{4}}=-1 $. From Corollary~\ref{coro:weofcode}, the code $ C_D $ has parameters $ [p^{m-2}+p^{\frac{m-3}{2}},m,(p-1) p^{m-3} ] $. Then
	\begin{align*}
	\sum_{i=0}^{m-1} \ceil{\frac{d}{p^i}} 
	& = p^{m-2}+1.
	\end{align*}
	So the equation $ p^{m-2}+1=  p^{m-2}+p^{\frac{m-3}{2}} $ gives that
	$ m=3 $, consequently $ \eta(-3)=-1 $. Thus $ C_D $ is MDS with parameters $ [p+1,3,p-1] $.
	
	Suppose that $ m_p \neq 0 $ and $ \eta(-m_p)(-1)^{\frac{(p-1)(m+1)}{4}}=1 $. In the same manner we obtain that $ C_D $ is MDS with parameters $ [p-1,3,p-3]  $ when $ m=3 $ and $ \eta(-3)=1 $.
	
	Hence we conclude that $ C_D $ is an optimal code achieving
	the Griesmer bound by Lemma~\ref{lem:Grie} provided that $ m=3 $.
\end{proof}

In the following, we present some examples to illustrate our results.
\begin{example}
	Let $(p,m)=(3,6)$. This corresponds to the case that $ 2\mid m $ and $ m_p=0 $. By Theorem~\ref{thm:cweofcode}, the code $C_{D}$ is an $[81, 6, 48]$ linear code. Its complete weight enumerator is
	\begin{align*}
	&z_0^{81} + z_1^{81} + z_2^{81} + 162 z_0^{33}z_1^{24}z_2^{24 }
	+ 240 ( z_0 z_1 z_2)^{27 }
	+162 z_0^{24}z_1^{24}z_2^{33 } +
	162 z_0^{24}z_1^{33}z_2^{24 }, 	
	\end{align*}
	and its weight enumerator is
	\begin{equation*}
	1+162x^{48}+240x^{ 54} + 324 x^{ 57}+2x^{81} .
	\end{equation*}
	These results coincide with numerical computation by Magma.

	Let us define the code $ C_{D_1} $ of~\eqref{def:CD} where
	\begin{align*}
	D_1 =\left\{ x\in \mathbb{F}_{r} :  \mathrm{Tr}(x)=1 \right\}.
	\end{align*}
	Magma works out that $ C_{D_1} $ has parameters $ [243,6,162] $. It is clear that the rate of $ C_{D } $ is higher than that of $ C_{D_1} $.
\end{example}
\begin{example}
	Let $(p,m)=(5,4)$. This corresponds to the case that $ 2\mid m $ and $ m_p \neq 0 $. By Theorem~\ref{thm:cweofcode}, the code $C_{D}$ is a $[ 20,4,14]$ linear code. Its complete weight enumerator is
	\begin{align*}
	&z_0^{ 20}+ z_1^{ 20}+ z_2^{ 20} + z_3^{ 20}+ z_4^{ 20}
	+ 20   z_1 z_2 (z_0 z_3 z_4)^{6}
	+ 20  z_1 z_3 (z_0 z_2 z_4)^{6}
	+ 20  z_1 z_4 (z_0 z_2 z_3)^{6}\\
	&+ 20  z_2 z_3 (z_0 z_1 z_4)^{6}
	+ 20  z_2z_4 (z_0 z_1z_3)^{6}
	+ 20  z_3z_4 (z_0 z_1z_2)^{6}
	+ 24 (z_0 z_2 z_3z_4)^{5}
	+ 24 (z_0 z_1   z_3z_4)^{5}\\	
	&+ 24 (z_0 z_1 z_2  z_4)^{5}
	+ 24 (z_0 z_1 z_2 z_3 )^{5}
	+ 300 (z_0 z_1z_2z_3z_4 )^{4}
	+ 20 z_0z_1(z_2z_3z_4  )^{6}
	+ 20 z_0z_2(z_1z_3z_4  )^{6}\\
	&+ 20 z_0z_3(z_1z_2z_4  )^{6}
	+ 20 z_0z_4(z_1z_2z_3  )^{6}
	+ 24 ( z_1 z_2 z_3 z_4 )^{5}
	,
	\end{align*}
	and its weight enumerator is
	\begin{equation*}
	1+120x^{14 }+96 x^{15} +300x^{ 16 }+80 x^{19}+28x^{20}.
	\end{equation*}
	These results coincide with numerical computation by Magma.
	This code is optimal according to Markus Grassl's table (see
	http://www.codetables.de/).
	
	Let us define the code $ C_{D_2} $ of~\eqref{def:CD} with the defining set
	\begin{align*}
	D_2 =\left\{ x\in \mathbb{F}_{r}^*  :  \mathrm{Tr}(x^2)=0 \right\}.
	\end{align*}
	The code $ C_{D_2} $ has been studied in~\cite{ding2015twothree}, which has parameters $ [104, 4, 80] $. It is clear that the rate of $ C_{D } $ is higher than that of $ C_{D_2} $.
	
\end{example}

\begin{example}
	Let $(p,m)=(5,3)$. This corresponds to the case that $ 2\nmid m $ and $ m_p \neq 0 $. We have $ \eta(-3)=-1 $. By Corollary~\ref{coro:MDS}, the code $C_{D}$ is MDS with parameters $[6,3,4]$. Its complete weight enumerator is
	\begin{align*}
	&z_0^{6}+ z_1^{ 6}+ z_2^{ 6} + z_3^{ 6} + z_4^{ 6}
	+ 6  z_1 z_4( z_2 z_3)^{2}
	+ 6  (z_1 z_4)^{2} z_2 z_3
	+ 6  (z_1 z_2 z_3)^{2} 
	+ 6  (z_1 z_2 z_4)^{2} \\
	&+ 6  (z_1 z_3 z_4)^{2}
	+ 6  (z_2 z_3 z_4)^{2}  
	+ 6  (z_0 z_1 z_2)^{2}
	+ 6  (z_0 z_1 z_3)^{2}
	+ 6  (z_0 z_1 z_4)^{2}
	+ 6  (z_0 z_2 z_3)^{2}\\
	&+ 6  (z_0 z_2 z_4)^{2}
	+ 6  (z_0 z_1 )^{2} z_2 z_4
	+ 6  (z_0 z_2 )^{2} z_3 z_4
	+ 6  (z_0 z_3 )^{2} z_1 z_2
	+ 6  (z_0 z_4 )^{2} z_1 z_3\\
	&+ 6  (z_0 z_3 z_4)^{2}
	+ 6   z_0 z_1  (z_2 z_4)^{2}
	+ 6   z_0 z_2  (z_3 z_4)^{2}
	+ 6   z_0 z_3  (z_1 z_2)^{2}
	+ 6  z_0 z_4   (z_1 z_3)^{2} ,  	
	\end{align*}
	and its weight enumerator is
	\begin{align*}
	1+60x^{4 }+24 x^{5} +40  x^{ 6}.
	\end{align*}
	These results coincide with numerical computation by Magma.

\end{example}

\section{Concluding remarks}

We have constructed a class of linear codes with a few weights by giving two restrictions in the defining set. In particular, we obtain MDS codes. Moreover, the codes defined in this paper may have shorter length and higher information rate. So they can be employed to construct authentication codes using the framework of~\cite{ding2007generic} and~\cite{Ding2005auth} and the complete weight distributions of the codes allow the determination of the success  probability with respect to certain attacks. More codes can be constructed in this way and we leave this for future work.




\begin{thebibliography}{10}
	\expandafter\ifx\csname url\endcsname\relax
	\def\url#1{\texttt{#1}}\fi
	\expandafter\ifx\csname urlprefix\endcsname\relax\def\urlprefix{URL }\fi
	\expandafter\ifx\csname href\endcsname\relax
	\def\href#1#2{#2} \def\path#1{#1}\fi
	
	\bibitem{VPless2003funda}
	W.~C. Huffman, V.~Pless, {Fundamentals of Error-Correcting Codes}, Cambridge
	University Press, Cambridge, 2003.
	
	\bibitem{macwilliams1977theory}
	F.~J. MacWilliams, N.~J.~A. Sloane, The Theory of Error-Correcting Codes,
	Vol.~16, North-Holland Publishing, Amsterdam, 1977.
	
	\bibitem{Blake1991}
	I.~F. Blake, K.~Kith, {On the complete weight enumerator of Reed-Solomon
		codes}, SIAM J. Discret. Math. 4~(2) (1991) 164--171.
	
	\bibitem{kith1989complete}
	K.~Kith, {Complete weight enumeration of Reed-Solomon codes}, Master's thesis,
	Department of Electrical and Computing Engineering, University of Waterloo.
	
	\bibitem{kuzmin1999complete}
	A.~Kuzmin, A.~Nechaev, {Complete weight enumerators of generalized Kerdock code
		and linear recursive codes over Galois ring}, in: Workshop on Coding and
	Cryptography, 1999, pp. 332--336.
	
	\bibitem{kuzmin2001complete}
	A.~Kuzmin, A.~Nechaev, {Complete weight enumerators of generalized Kerdock code
		and related linear codes over Galois ring}, Discrete Applied Mathematics 111
	(2001) 117--137.
	
	\bibitem{Nebe2004}
	G.~Nebe, H.~Quebbemann, E.~Rains, N.~Sloaned, Complete weight enumerators of
	generalized doubly-even self-dual codes, Finite Fields and Their Applications
	10 (2004) 540--550.
	
	\bibitem{helleseth2006}
	T.~Helleseth, A.~Kholosha, Monomial and quadratic bent functions over the
	finite fields of odd characteristic, IEEE Transactions on Information Theory
	52~(5) (2006) 2018--2032.
	
	\bibitem{ding2007generic}
	C.~Ding, T.~Helleseth, T.~Kl\o{}ve, X.~Wang, {A generic construction of
		Cartesian authentication codes}, IEEE Transactions on Information Theory
	53~(6) (2007) 2229--2235.
	
	\bibitem{Ding2005auth}
	C.~Ding, X.~Wang, A coding theory construction of new systematic authentication
	codes, Theoretical Computer Science 330 (2005) 81--99.
	
	\bibitem{chu2006constantco}
	W.~Chu, C.~J. Colbourn, P.~Dukes, On constant composition codes, Discrete
	Applied Mathematics 154 (2006) 912--929.
	
	\bibitem{ding2008optimal}
	C.~Ding, Optimal constant composition codes from zero-difference balanced
	functions, IEEE Transactions on Information Theory 54~(12) (2008) 5766--5770.
	
	\bibitem{ding2006construction}
	C.~Ding, J.~Yin, A construction of optimal constant composition codes, Designs,
	Codes and Cryptography 40 (2006) 157--165.
	
	\bibitem{ding2015twodesign}
	C.~Ding, Linear codes from some 2-designs, IEEE Transactions on Information
	Theory 61~(6) (2015) 3265--3275.
	
	\bibitem{dingkelan2014binary}
	K.~Ding, C.~Ding, Binary linear codes with three weights, IEEE Communications
	Letters 18~(11) (2014) 1879--1882.
	
	\bibitem{ding2015twothree}
	K.~Ding, C.~Ding, A class of two-weight and three-weight codes and their
	applications in secret sharing, IEEE Transactions on Information Theory
	61~(11) (2015) 5835--5842.
	
	\bibitem{LiYang2015cwe}
	C.~Li, S.~Bae, J.~Ahn, S.~Yang, Z.~Yao, Complete weight enumerators of some
	linear codes and their applications, Designs, Codes and Cryptography 81
	(2016) 153--168.
	
	\bibitem{yang2015linear}
	S.~Yang, Z.~Yao, Complete weight enumerators of some linear codes,
	arXiv:1505.06326.
	
	\bibitem{AhnKaLi2016completegenelize}
	J.~Ahn, D.~Ka, C.~Li, Complete weight enumerators of a class of linear codes,
	Designs, Codes and Cryptography\href
	{http://dx.doi.org/10.1007/s10623-016-0205-8}
	{\path{doi:10.1007/s10623-016-0205-8}}.
	
	\bibitem{li2015complete}
	C.~Li, Q.~Yue, F.-W. Fu, Complete weight enumerators of some cyclic codes,
	Designs, Codes and Cryptography 80 (2016) 295--315.
	
	\bibitem{WangQiuyan2015complet}
	Q.~Wang, F.~Li, K.~Ding, D.~Lin, Complete weight enumerators of two classes of
	linear codes, arXiv:1512.07341.
	
	\bibitem{wang2015class}
	Q.~Wang, F.~Li, D.~Lin, A class of linear codes with three weights,
	arXiv:1512.03866.
	
	\bibitem{yang2016compthree}
	S.~Yang, Z.~Yao, C.~Zhao, A class of three-weight linear codes and their
	complete weight enumerators, Cryptography and Communications\href
	{http://dx.doi.org/10.1007/s12095-016-0187-4}
	{\path{doi:10.1007/s12095-016-0187-4}}.
	
	\bibitem{yang2015complete}
	S.~Yang, Z.~Yao, Complete weight enumerators of a family of three-weight linear
	codes, Designs, Codes and Cryptography\href
	{http://dx.doi.org/10.1007/s10623-016-0191-x}
	{\path{doi:10.1007/s10623-016-0191-x}}.
	
	\bibitem{Yang2016complete}
	S.~Yang, Z.~Yao, Complete weight enumerators of a class of linear codes,
	Discrete Mathematics\href
	{http://dx.doi.org/http://dx.doi.org/10.1016/j.disc.2016.11.029}
	{\path{doi:http://dx.doi.org/10.1016/j.disc.2016.11.029}}.
	
	\bibitem{Lifei15wt}
	F.~Li, Q.~Wang, D.~Lin, A class of three-weight and five-weight linear codes,
	arXiv:1509.06242.
	
	\bibitem{Heng2016}
	Z.~Heng, Q.~Yue, Two classes of two-weight linear codes, Finite Fields and
	Their Applications 38 (2016) 72 -- 92.
	
	\bibitem{calderbank1984three}
	A.~Calderbank, J.~Goethals, Three-weight codes and association schemes, Philips
	J. Res 39.
	
	\bibitem{storer1967cyclotomy}
	T.~Storer, {Cyclotomy and Difference Sets}, Markham Publishing Company,
	Markham, Chicago, 1967.
	
	\bibitem{lidl1983finite}
	R.~Lidl, H.~Niederreiter, Finite fields, Encyclopedia of Mathematics and its
	Applications. Reading, Massachusetts, USA: Addison-Wesley 20.
	
	\bibitem{Griesmer1960}
	J.~H. Griesmer, A bound for error-correcting codes, IBM Journal of Research and
	Development 4~(5) (1960) 532--542.
	
\end{thebibliography}

\end{document}